\documentclass[11pt]{article}

\usepackage{ifthen}
\def\lncs{0}
\ifthenelse{\equal{\lncs}{1}}
{
\usepackage{amsmath,amssymb,xspace,verbatim,paralist,enumitem}
\usepackage[usenames, dvipsnames]{xcolor}
\usepackage{array,multirow,multicol}
\usepackage{algorithm,algorithmicx}
\usepackage[noend]{algpseudocode}
\usepackage{graphicx}
\usepackage{mathdots,mathtools,mathrsfs}
\usepackage{booktabs}

\usepackage{thmtools}
\usepackage{thm-restate}

\usepackage[normalem]{ulem} 

\setlist[enumerate]{noitemsep}

\usepackage[pagebackref, linktoc=all,colorlinks=true,pdfpagemode=none,urlcolor=blue,linkcolor=blue,citecolor=violet,pdfstartview=FitH]{hyperref}
\usepackage[nameinlink]{cleveref}

\usepackage{subcaption}

\newbox\mybox 
\newdimen\myboxwidth    

\newcommand\addpicture[3]{%
\setbox\mybox=\hbox{\includegraphics[scale=#3]{#2}}
\myboxwidth\wd\mybox    
\renewcommand\windowpagestuff{%
\includegraphics[scale=#3]{#2}
\captionof{figure}{A test figure.}}
\parpic[#1]{%
\begin{minipage}{\myboxwidth}
 \windowpagestuff 
\end{minipage} 
} }

\newcommand{\ignore}[1]{}

\newcommand{\cA}{{\cal A}}
\newcommand{\cB}{\mathcal{B}}

\newcommand{\cE}{{\cal E}}
\newcommand{\cF}{\mathcal{F}}

\newcommand{\cI}{{\cal I}}
\newcommand{\cJ}{{\cal J}}

\newcommand{\cX}{{\cal X}}

\newcommand{\R}{\mathbb R}

\newcommand{\eps}{\varepsilon}

\newcommand{\poly}{\mathrm{poly}}

\DeclarePairedDelimiter\ceil{\lceil}{\rceil}
\DeclarePairedDelimiter\floor{\lfloor}{\rfloor}

\DeclareMathOperator{\violate}{\lambda}
\DeclareMathOperator{\approximate}{\rho}
\DeclareMathOperator{\opt}{OPT}

\newcommand{\etal}{{et al.}\xspace}

\newcommand{\mbdmb}{\textsc{MBDMB}\xspace}

\newcommand{\clu}{\textsc{Vanilla $(k , p)$-Clustering}\xspace}
\newcommand{\fairClu}{\textsc{Fair $(k , p)$-Clustering}\xspace}
\newcommand{\fairCluAsgn}{\textsc{Fair $p$-Assignment}\xspace}

\newcommand{\lbclu}{\textsc{LB-$(k , p)$-Clustering}\xspace}

\newcommand{\optv}{\opt_{\mathsf{vnll}}}
\newcommand{\optf}{\opt_{\mathsf{fair}}}
\newcommand{\opta}{\opt_{\mathsf{asgn}}}
\newcommand{\optl}{\opt_{\mathsf{lbnd}}}
\newcommand{\near}{\mathsf{nrst}}

\newcommand{\zv}{\vec{d^*}}
\newcommand{\zf}{\vec{d}}
\newcommand{\zs}{\vec{d'}}
\newcommand{\pnorm}{\mathcal{L}_p}

\newcommand{\LP}{\mathsf{LP}}
\newcommand{\veca}{\vec{\alpha}}
\newcommand{\vecb}{\vec{\beta}}

\DeclareMathOperator{\balance}{balance}

\newcommand\Tstrut {\rule{0pt}{3ex}}         
\newcommand\Bstrut {\rule[-1.3ex]{0pt}{0pt}}   
}
{
\usepackage{amsthm}
\usepackage{fullpage}
\usepackage{times}
\usepackage{footmisc}
 
\newtheorem{theorem}{Theorem}
\theoremstyle{definition}
\newtheorem{definition}{Definition}
\newtheorem{lemma}[theorem]{Lemma}
\newtheorem{claim}[theorem]{Claim}

\newtheorem{remark}{Remark}
}

\usepackage{graphicx}
\begin{document}
\title{Fair Algorithms for Clustering}

\ifthenelse{\equal{\lncs}{1}}{
\author{Suman K. Bera\inst{1} \and
Deeparnab Chakrabarty\inst{1} \and
Maryam Negahbani\inst{1}}
\authorrunning{S.K. Bera et al.}
%
\institute{Dartmouth College, Hanover NH 03755, USA}
}
{
\author{Suman K. Bera\footnote{UC Santa Cruz, Email: sbera@ucsc.edu} \and Deeparnab Chakrabarty\footnote{Dartmouth College, Email: deeparnab@dartmouth.edu}
\and Nicolas J. Flores\footnote{Dartmouth College, Email: nicolasflores.19@dartmouth.edu}
\and Maryam Negahbani\footnote{Dartmouth College, Email: maryam@cs.dartmouth.edu}}
}
\date{}

\maketitle              

\begin{abstract}

We study the problem of finding low-cost {\em fair clusterings} in data where each data point  may 
belong to many protected groups. 
Our work significantly generalizes the seminal work of Chierichetti \etal (NIPS 2017) as follows.
\begin{itemize}[noitemsep]
    \item We allow the user to specify the parameters that define fair representation. More precisely, these parameters define the maximum over- and minimum under-representation of any group in any cluster.
    \item Our clustering algorithm works on any $\ell_p$-norm objective (e.g. $k$-means, $k$-median, and $k$-center). Indeed, our algorithm transforms any vanilla clustering solution into a fair one incurring only a slight loss in quality.
    \item Our algorithm also allows individuals to lie in multiple protected groups. 
    In other words, we do not need the protected groups to partition the data and we can maintain fairness across different groups simultaneously.
\end{itemize}
Our experiments show that on established data sets, our algorithm performs much better in practice than what our theoretical results suggest. 

\end{abstract}

\section{Introduction}

Many important decisions today are made by machine learning algorithms. These range from showing advertisements to customers~\cite{perlich2014machine,googlead}, to awarding home loans~\cite{khandani2010consumer,malhotra2003evaluating}, to predicting recidivism~\cite{propublica,dressel2018accuracy,chouldechova2017fair}. It is important to ensure that such algorithms are {\em fair} and are not biased towards or against some specific groups in the population. A considerable amount of work~\cite{kamishima2012fairness,zemel2013learning, CKLV18,joseph2016fairness,celis2018,ZafarVGG17,yang2017measuring} addressing this issue has emerged in the recent years.

Our paper considers fair algorithms for clustering. Clustering is a fundamental unsupervised learning problem where one wants to partition a given data-set. 
In machine learning, clustering is often used for feature generation and enhancement as well. It is thus important to consider the bias and unfairness issues when inspecting the quality of clusters.
%
The question of fairness in clustering was first asked in the beautiful paper of Chierichetti~\etal~\cite{CKLV18} with subsequent generalizations by R\"osner and Schmidt~\cite{RS18}. 

	\emph{
		In this paper, we give a much more generalized and tunable notion of fairness in clustering than that in~\cite{CKLV18,RS18}.
		Our main result is that any solution for a wide suite of vanilla clustering objectives can be transformed into fair solutions in our notion with only a slight loss in quality by a simple algorithm. 
	}


Many works in fairness~\cite{celis2018,CKLV18,RS18,celis2018multiwinner} work within the disparate impact (DI) doctrine~\cite{feldman2015certifying}. Broadly speaking, the doctrine posits that any ``protected class'' must have approximately equal representation in the decisions taken (by an algorithm).
Although the DI doctrine is a law~\cite{401,EEOC} in the United States, violating the DI doctrine is by itself {\em not} illegal~\cite{AC05}; it is illegal only if the violation cannot be justified by the decision maker. 
 In the clustering setting, this translates to the following algorithmic question : what is the loss in quality of the clustering when all protected classes are required to have approximately equal representation in the clusters returned?


Motivated thus, Chierichetti~\etal~\cite{CKLV18}, and later R\"osner and Schmidt~\cite{RS18}, model the set of points as partitioned into $\ell$ colors, and the color proportion of each returned cluster should be similar to that in the original data. There are three shortcomings of these papers: (a) the fairness constraint was too stringent and brittle, (b) good algorithms were given only for the $k$-center objective, and (c) the color classes weren't allowed to overlap. We remark that the last restriction is limiting since an individual can lie in multiple protected classes (consider an African-American senior woman). In our work we address all these concerns: we allow the user to specify the fairness constraints, we give simple algorithms with provable theoretical guarantees for a large suite of objective functions, and we allow overlapping protected classes.

{\bf Our fairness notion.} 
We propose a model which extends the model of~\cite{CKLV18} to have $\ell \ge 2$  groups of people which are allowed to overlap.
For each group $i$, we have two parameters $\beta_i,\alpha_i \in [0,1]$.
Motivated by the DI doctrine, we deem a clustering solution {\em fair} if each cluster satisfies two properties: (a) {\em restricted dominance (RD)}, which asserts that the fraction of people from group $i$ in any cluster is at most $\alpha_i$, and (b) {\em minority protection (MP)}, which asserts that the fraction of people from group $i$ in any cluster is at least $\beta_i$.
Note that we allow $\beta_i,\alpha_i$'s to be arbitrary parameters, and furthermore, they can differ across different groups.
This allows our model to provide a lot of flexibility to users. For instance, our model easily captures
the notions defined by~\cite{CKLV18} and~\cite{RS18}. 

We allow our protected groups to overlap. Nevertheless, the quality of our solutions depend on the amount of overlap. We define $\Delta$ 
 (similar to~\cite{celis2018}) to be the maximum number of groups a single individual can be a part of. This parameter, as we argued above, is usually not $1$, but can be assumed to be a small constant depending on the application.


{\bf Our results.}
Despite the generality of our model, we show that \emph{in a black-box fashion}, we can get fair algorithms for {\em any} $\ell_p$-norm objective (this includes, $k$-center, $k$-median, and the widely used $k$-means objective) if we allow for very small additive violations to the fairness constraint. 
We show that \emph{given any $\rho$-approximation algorithm $\cA$ for a given objective which could be returning widely unfair clusters, we can return a solution which is a $(\rho+2)$-approximation to the best clustering which satisfies the fairness constraints} (\Cref{thm:fairclu}). Our solution, however, can violate both the RD and MP property {\em additively} by $4\Delta+3$. This is negligible if the clusters are large, and our empirical results show this almost never exceeds $3$. Further in our experiments, our cost is at most 15\% more than optimum, which is a much better factor compared to $(\rho+2)$.
%

The black-box feature of our result is useful also in {\em comparing} the performance of any particular algorithm $\cA$. This helps if one wishes to justify the property of an algorithm one might be already using. Our results can be interpreted to give a way to convert any clustering algorithm to its fair version. Indeed, our method is very simple -- we use the solution returned by $\cA$ to define a {\em fair assignment} problem and show that this problem has a good optimal solution. The fair assignment problem is then solved via iterative rounding which leads to the small additive violations. In the case of $\Delta=1$ (disjoint groups), we can get a simpler, one-iteration rounding algorithm.



Finally, we show that our simple approach also leads to algorithms for a related clustering problem.
In many clustering applications involving anonymity and privacy~\cite{AggarwalFKKPTZ06,RS18}, one requires the size of the cluster to be {\em at least} a certain size $L$. We show that \emph{given any $\rho$-approximation for the vanilla clustering problem in any $\ell_p$ norm, we can get a $(\rho+2)$-approximation for the lower bounded clustering problem in $O(2^k\poly(n))$ time} (\Cref{thm:lbfpt}). Thus, our algorithm is a {\em fixed-parameter tractable (FPT)} approximation algorithm and in particular, armed with the recent result of Cohen-Addad \etal \cite{CohenGKLL19}, implies a $3.736$-factor approximation algorithm for the lower bounded $k$-median problem in $k^{O(k)}\poly(n)$ time. 
To put this in perspective, in polynomial time one can only get a large $O(1)$-approximation (see~\cref{fn:lb}).
Furthermore, for higher norms, no constant factor approximations are known.\footnote{For the special case of Euclidean $k$-means, there are PTASes in \cite{DX15,BJK18} with run times exponential in $k$.}

\textbf{Comparison with recent works.}
In a very recent independent and concurrent work, Schmidt \etal ~\cite{schmidt2018fair}
consider the fair $k$-means problem in the
streaming model with a notion of fairness similar to ours.
However, their results crucially assume that the underlying 
metric space is Euclidean.
Their main contributions  are defining ``fair coresets'' 
and showing how to compute them in a streaming setting,
resulting in significant reduction in the input size. 
Although their coreset construction algorithm works with arbitrary
number of groups, their fair $k$-means algorithms 
assume there are only two disjoint groups of equal size. 
Even for this, 
Schmidt \etal ~\cite{schmidt2018fair} give an $(5.5\rho+1)$-approximation,
given any $\rho$-approximation 
for the vanilla $k$-means problem; the reader should compare with our $(\rho+2)$-approximation.
Backurs \etal~\cite{Backurs2019}
consider the problem of designing scalable 
algorithm for the fair k-median problem in the Euclidean space. 
The notion of fairness is {\em balance}, as defined by Chierichetti
\etal~\cite{CKLV18}, and hence works only for two disjoint groups.
Their approximation ratio is $O_{r,b}(d\log n)$ where $r$ and $b$ are 
fairness parameters, and $d$ is the dimension of the Euclidean space.
In contrast, our fair $k$-means and $k$-median algorithms works in 
any metric space, with arbitrary number of overlapping groups. 

In another independent and parallel work, 
Bercea \etal~\cite{Bercea2018} consider
a fairness model that is similar to ours. They give a similar, 
but arguably more complicated algorithm for a variety of 
clustering objectives. Ahmadian \etal~\cite{sara2019clustering}
study the k-center objective with 
only {\em restricted dominance (RD)} type constraints and give
bi-criteria approximations.
In comparison, we emphasize on a 
simple, yet powerful unifying framework that can handle
any $\ell_p$-norm objective.
None of the above works handle overlapping groups.

\subsection{Other related works}\label{sec:relwork}
Fairness in algorithm design has received a lot of 
attention lately~\cite{calders2010three,luong2011k,dwork2012fairness,feldman2015certifying,
	kamishima2012fairness,zemel2013learning,CKLV18,joseph2016fairness,celis2018,ZafarVGG17,yang2017measuring,celis2018,celis2018multiwinner,corbett2017algorithmic,KleinbergMR17,FriedlerSV16}.
Our work falls in the category of designing fair algorithms, and as mentioned, we 
concentrate on the notion of {\em disparate impact}. 
Feldman \etal~\cite{feldman2015certifying} and Zafar \etal~\cite{ZafarVGG17} 
study the fair classification problem under this notion. 
Celis~\etal in ~\cite{celis2018}, Celis~\etal in~\cite{celis2018multiwinner}, and Chierichetti \etal in~\cite{chierichetti19a}
study respectively the fair ranking problem,
the multiwinner voting problem, and the matroid optimization problem; All of these works model fairness through {\em disparate impact}. 
Chierichetti \etal in~\cite{CKLV18} first addresses {\em disparate impact} for clustering problems
in the presence of two groups, R{\"o}sner and Schmidt~\cite{RS18} 
generalizes it to more than two groups.

Chen \etal~\cite{chen2019proportionally} define a notion of 
proportionally fair clustering where all possible groups of 
reasonably large size are entitled to choose a center for themselves.
This work builds on the assumption that sometimes the task of identifying
protected group itself is untenable. Kleindessner \etal in~\cite{Kleindessner2019} study the problem of enforcing fair 
representation in the data points chosen as cluster center. This problem
can also be posed as a matroid center problem. Kleindessner \etal 
in~\cite{Kleindessner2019spectral} extends the fairness notion to 
graph spectral clustering problems. Celis~\etal in~\cite{Celis:2019}
proposes a meta algorithm for the classification problem under 
a large class of fairness constraints with respect to 
multiple non-disjoint protected groups.

Clustering is a ubiquitous problem and has been extensively studied in diverse communities (see~\cite{aggarwal2013data} for a recent survey). We focus on the work done in the algorithms and optimization community for clustering problems under $\ell_p$ norms.
%
The $p=\{1,2,\infty\}$ norms, that is the $k$-median, $k$-means, and $k$-center problems respectively, have been extensively studied.
The $k$-center problem has a $2$-approximation~\cite{HochbaumS85a,Gonzalez85} and it is NP-hard to do better~\cite{hsu1979easy}.
A suite of algorithms~\cite{CharikarGST02,JainV01,CharikarG99,AryaGKMMP04,LiS16} for the $k$-median problem has culminated in a 
$2.676$-approximation~\cite{ByrkaPRS14}, and is still an active area of research. For $k$-means, the best algorithm is a $9+\eps$-approximation due to Ahmadian~\etal~\cite{AhmadianNSW17}. For the general $p$-norm, most of the $k$-median algorithms imply a constant approximation.

Capacitated clustering is similar to fair clustering in that in both, the assignment is not implied by the set of centers opened.
We already mentioned the results for {\em lower bounded} clustering. One can also look at {\em upper bounded} clustering where every cluster is at most a size $U$. The (upper-bounded) {\em capacitated $k$-median} problem is one of the few classic problems remaining for which we do not know $O(1)$-approximations, and neither we know of a good hardness. The capacitated $k$-center problem has a $6$-approximation~\cite{KS00}. Recently, an FPT algorithm was designed by~\cite{Adamczyk18}; They show a $7+\eps$-approximation for the upper bounded capacitated $k$-median problem which runs in time $O(f(k)\cdot \poly(n))$ where $f(k) \sim k^{O(k)}$. It is instructive to compare this with our result on lower bounded $k$-median problem.

\section{Preliminaries}
\label{sec:prelim}
%
Let $C$ be a set of points (whom we also call ``clients'') we want to cluster.
Let these points be embedded in a metric space $(\cX,d)$. We let $F \subseteq \cX$ be the set of possible cluster center locations (whom we also call ``facilities''). Note $F$ and $C$ needn't be disjoint, and indeed $F$ could be equal to $C$.
For a set $S\subseteq \cX$ and a point $x\in \cX$, we use $d(x,S)$ to denote $\min_{y\in S} d(x,y)$.
For an integer $n$, we use $[n]$ to denote the set $\{1,2,\ldots,n\}$.

	Given the metric space $(\cX,d)$ and an integer parameter $k$, in the \clu problem  
	the objective is to (a) {\em ``open''} a subset $S\subseteq F$ 
	of at most $k$ facilities, and (b) find an {\em assignment} $\phi: C\to S$ of clients to open facilities so as to minimize  $
    \pnorm(S;\phi) := \Big(\sum_{v\in C} d(v,\phi(v))^p\Big)^{\frac{1}{p}}.$ Indeed, in this vanilla version with no fairness considerations, every point $v\in C$ would be assigned to the closest center in $S$.
    The case of $p=\{1,2,\infty\}$, the $k$-median, $k$-means, and $k$-center problems respectively, have been extensively studied in the literature~\cite{HochbaumS85a,Gonzalez85,CharikarGST02,JainV01,CharikarG99,AryaGKMMP04,LiS16,ByrkaPRS14,AhmadianNSW17}. Given an instance $\cI$ of the \clu problem, we use $\optv(\cI)$ to denote its optimal value.

\noindent
The next definition formalizes the fair clustering problem which is the main focus of this paper.

\begin{definition}[\fairClu Problem]\label{def:fairClu}
	In the fair version of the clustering problem, one is additionally given 
	$\ell$ many (not necessarily disjoint) {\em groups} of $C$, namely $C_1,C_2,\ldots,C_{\ell}$. 
	We use $\Delta$ to denote the maximum number of groups a single client $v\in C$ can belong to; so if the $C_j$'s were disjoint we would have $\Delta = 1$.
	One is also given two {\em fairness vectors} $\veca, \vecb \in [0,1]^{\ell}$.

The objective is to (a) \emph{open} a subset of facilities $S \subseteq F$ of at most $k$ facilities, and (b) find an {\em assignment} $\phi: C\to S$ of clients to the open facilities so as to minimize
$\pnorm(S;\phi)$, where $\phi$ satisfies the following {\em fairness constraints}.
\begin{alignat}{4}
 \Big| \{v \in C_i : \phi(v)=f \}\Big| && ~~\leq~~ &&  \alpha_i \cdot \Big|\{v \in C: \phi(v)=f \}\Big|, && \quad
& \forall f\in S\,, \forall i \in [\ell] \,, \label{eqn:fairclu:majority} \tag{RD} \\
 \Big| \{v \in C_i : \phi(v)=f \}\Big| && ~~\geq~~ &&  \beta_i \cdot \Big|\{v \in C: \phi(v)=f \}\Big|, && \quad
& \forall f\in S\,, \forall i \in [\ell] \,, \label{eqn:fairclu:minority} \tag{MP} 
\end{alignat}
The assignment $\phi$ defines a cluster $\{v:\phi(v)=f\}$ around every open facility $f\in S$.
As explained in the Introduction,~\cref{eqn:fairclu:majority} is the {\em restricted dominance} property which upper bounds the ratio of any 
group's participation in a cluster, and~\cref{eqn:fairclu:minority} is the {\em minority protection} property which lower bounds this ratio to protect against under-representation.
Due to these fairness constraints, we can no longer  assume $\phi(v)$ is the nearest open facility in $S$ to $v$. Indeed, we use the tuple $(S,\phi)$ to denote a fair-clustering solution.

We use $\optf(\cI)$ to denote the optimal value of any
instance $\cI$ of the \fairClu problem. Since $\cI$ is also an instance of the vanilla problem, and since every fair solution is also a vanilla solution (but not necessarily vice versa) 
we get 
$\optv(\cI) \leq \optf(\cI)$ for any $\cI$.
%
%
\end{definition}
%


\noindent
A fair clustering solution $(S,\phi)$ has $\lambda$-{\em additive} violation, if the~\cref{eqn:fairclu:majority} and~\cref{eqn:fairclu:minority} constraints are satisfied upto $\pm \lambda$-violation. More precisely, for any $f\in S$ and for any group $i\in [\ell]$, we have 
	\begin{equation}\label{eq:viol}
	 \beta_i \cdot \Big|\{v \in C: \phi(v)=f \}\Big| - \lambda \le \Big| \{v \in C_i : \phi(v)=f \}\Big| \leq  \alpha_i \cdot \Big|\{v \in C: \phi(v)=f \}\Big| + \lambda \tag{V}
	\end{equation}

\noindent
Our main result is the following.
\begin{restatable}{theorem}{ThmfairkCent}
	\label{thm:fairclu}
	Given a $\approximate$-approximate algorithm $\cA$ for the \clu problem,
	we can return a $(\approximate+2)$-approximate solution $(S,\phi)$ with $(4\Delta+3)$-additive violation for the \fairClu problem.
\end{restatable}

In particular, we get $O(1)$-factor approximations to the \fairClu problem with $O(\Delta)$ {\em additive} violation, for any $\ell_p$ norm\footnote{We cannot find an explicit reference for a vanilla $O(1)$-approximation for general $\ell_p$ norms, but it is not too hard to see that many $k$-median algorithms such as those of~\cite{CharikarGST02} and~\cite{JainV01} imply $O(1)$-approximations for any norm. The only explicit mention of the $p$-norm clustering we could find was the paper~\cite{GuptaT08}; this shows that {\em local search} gives a $\Theta(p)$ approximation.\label{fn:pn}}.
Furthermore, for the important special case of $\Delta = 1$, our additive violation is at most $+3$.

Our technique also implies algorithms for the {\em lower bounded $k$-clustering problem}. In this, there is no fairness constraint; rather, the constraint is that if a facility is opened, {\em at least} $L$ clients must be assigned to it. 
In this problem also, a client need not be assigned to its nearest open facility. The problem has been studied in the facility location version and (large) $O(1)$-factor algorithms are known. 
We can easily get  the following.
\begin{restatable}{theorem}{ThmLB}\label{thm:lbfpt}
	Given a $\approximate$-approximate algorithm $\cA$ for the \clu problem that runs in time $T$, there is a $(\approximate + 2)$-approximation algorithm for the \lbclu problem that 
	runs in time $O(T + 2^k\cdot\poly(n))$.
\end{restatable}
\noindent
In particular, armed with the recent result of Cohen-Addad \etal \cite{CohenGKLL19}, implies a $3.736$-factor approximation algorithm for the lower bounded $k$-median problem in $k^{O(k)}\poly(n)$ time. 

\section{Algorithm for the \fairClu problem}
\label{sec:algo}
Our algorithm is a simple two step procedure. 
First, we solve the \clu problem using some algorithm $\cA$, and fix the centers $S$ opened by
$\cA$. Then, we solve a {\em fair reassignment problem}, called 
\fairCluAsgn problem, on the same set of facilities to get assignment $\phi$. 
We return $(S,\phi)$ as our fair solution.

\begin{definition}[\fairCluAsgn Problem]
	In this problem, we are given the original set of clients $C$ and a set $S \subseteq F$ with $|S| = k$.
	The objective is to find the assignment $\phi:C\to S$ such that (a) the constraints~\cref{eqn:fairclu:majority} and~\cref{eqn:fairclu:minority} 
	are satisfied, and (b) $\pnorm(S;\phi)$ is minimized among all such satisfying assignments.
	\end{definition}
	Given an instance $\cJ$ of the \fairCluAsgn problem, we let $\opta(\cJ)$ denote its optimum value. 
	Clearly, given any instance $\cI$ of the \fairClu problem, if $S^*$ is the optimal subset for $\cI$ and $\cJ$ is the instance of \fairCluAsgn defined by $S^*$,
	then $\optf(\cI) = \opta(\cJ)$. A $\lambda$-violating algorithm for the \fairCluAsgn problem is allowed to incur $\lambda$-additive violation to the fairness constraints.

We present our algorithmic template for the \fairClu
problem in~\Cref{alg:fairclu}. This template uses the \textsc{FairAssignment}
procedure (~\Cref{algo:faircluassgn}) as a subroutine.
\begin{algorithm}[!ht]
    \caption{Algorithm for the \fairClu problem}
    \label{alg:fairclu}
    \begin{algorithmic}[1]
        \Procedure{FairClustering}{$(\cX=F\cup C,d)$, $C=\cup_{i=1}^{\ell}C_i$, $\veca, \vecb \in [0,1]^{\ell}$}
            \State solve the \clu problem on $(\cX,d)$
            \State let $(S, \phi)$ be the solution
            \State $\hat{\phi} = $ \textsc{FairAssignment} $((\cX,d), S,   C=\cup_{i=1}^{\ell}C_i, \veca, \vecb)$ (~\Cref{algo:faircluassgn})
            \State  return $(S,\hat{\phi})$
        \EndProcedure
    \end{algorithmic}
\end{algorithm}

\subsection{Reducing \fairClu to \fairCluAsgn}
\label{sec:reduction}
In this section we present a simple reduction from the
\fairClu problem to the \fairCluAsgn problem 
that uses a \clu solver as a black-box.

\begin{theorem}\label{fairclutoasgnthm}
Given a $\approximate$-approximate algorithm $\cA$ for the \clu problem and a $\violate$-violating algorithm $\cB$ for the \fairCluAsgn problem, there is a $(\approximate + 2)$-
approximation algorithm for the \fairClu problem with $\lambda$-additive violation.
\end{theorem}
\begin{proof}
 Given instance $\cI$ of the \fairClu problem,
we run $\cA$ on $\cI$ to get a (not-necessarily fair) solution $(S, \phi)$.
We are guaranteed $\pnorm(S;\phi) \leq \rho \cdot \optv(\cI) \leq \rho\cdot \optf(\cI)$.
Let $\cJ$ be the instance of \fairCluAsgn obtained by taking $S$ as the set of facilities.
We run algorithm $\cB$ on $\cJ$ to get a $\violate$-violating solution $\hat{\phi}$. We return $(S,\hat{\phi})$.

By definition of $\violate$-violating solutions, we get that $(S,\hat{\phi})$ satisfies~\cref{eq:viol} 
and that $\pnorm(S,\hat{\phi})\leq \opta(\cJ)$.
The proof of the theorem follows from the lemma below.
\end{proof}


\begin{lemma}\label{lemma:lowcost}
	$\opta(\cJ) \le \left(\rho+2\right)\cdot \optf(\cI)$.
\end{lemma}
\begin{proof}
Suppose the optimal solution of $\cI$ is $(S^*,\phi^*)$ with $\pnorm(S^*;\phi^*)=\optf(\cI)$. 
Recall $(S,\phi)$ is the solution returned by the $\rho$-approximate algorithm $\cA$.
We describe the existence of an assignment $\phi':C\to S$ such that $\phi'$ satisfies~\cref{eqn:fairclu:majority} and~\cref{eqn:fairclu:minority}, and
$\pnorm(S;\phi')\leq (\rho+2)\cdot \optf(\cI)$.
Since $\phi'$ is a feasible solution of $\cJ$, the lemma
follows.
For every $f^* \in S^*$, define $\near(f^*) := \arg\min_{f\in S} d(f,f^*)$ be the closest facility in $S$ to $f^*$.
For every client $v\in C$, define $\phi'(v) := \near(\phi^*(v))$
The following two claims prove the lemma. \qedhere
\end{proof}
\begin{claim}
\label{clm:fairness}
	$\phi'$ satisfies~\cref{eqn:fairclu:majority} and~\cref{eqn:fairclu:minority}
\end{claim}
\begin{proof}
For any facility $f^*\in S^*$, let $C(f^*) := \{v: \phi^*(v)=f^*\}$. The $C(f^*)$'s partition $C$.
For any $i\in [\ell]$, let $C_i(f^*) := C(f^*)\cap C_i$. Since $(S^*;\phi^*)$ is a feasible solution satisfying the fairness constraints,
we get that for every $f^*\in S^*$ and for every $i\in[\ell]$, $\beta_i\le \frac{|C_i(f^*)|}{|C(f^*)|}\le \alpha_i$.

For any facility $f\in S$, let $N(f) := \{f^*\in S^*: \near(f^*) =f\}$ be all the facilities in $S^*$ for which $f$ is the nearest facility.
Note that the clients $\{v\in C: \phi'(v)=f\}$ are precisely $\dot{\cup}_{f^*\in N(f)} C(f^*)$. Similarly, for any $i\in [\ell]$, 
we have $\{v\in C_i: \phi'(v)=f\}$ is precisely $\dot{\cup}_{f^*\in N(f)} C_i(f^*)$. Therefore, 
$
\frac{|\{v\in C_i: \phi'(v)=f\}|}{|\{v\in C: \phi'(v)=f\}|} = \frac{\sum_{f^*\in N(f)} |C_i(f^*)|}{\sum_{f^*\in N(f)} |C(f^*)|} \in [\beta_i,\alpha_i]
$
since the second summation is between $\min_{f^*\in N(f)} |C_i(f^*)|/|C(f^*)|$ and $\max_{f^*\in N(f)} |C_i(f^*)|/|C(f^*)|$, and both these are in $[\beta_i,\alpha_i]$.
\end{proof}

\begin{claim}
\label{clm:distance}
	$\pnorm(S;\phi')\leq (\rho+2)\optf(\cI)$.
\end{claim}
\begin{proof}
%
Fix a client
$v \in C$. For the sake of brevity, let: $f = \phi(v)$, $f' = \phi'(v)$, and $f^* = \phi^*(v)$. 
We have
\ifthenelse{\equal{\lncs}{1}}
{
    $d(v,f') = d(v,\near(f^*)) \leq d(v,f^*) + d(f^*,\near(f^*)) \leq d(v,f^*) + d(f^*,f) \leq 2d(v,f^*) + d(v,f)$.
}
{
\[
d(v,f') = d(v,\near(f^*)) \leq d(v,f^*) + d(f^*,\near(f^*)) \leq d(v,f^*) + d(f^*,f) \leq 2d(v,f^*) + d(v,f)
\]
}
The first and third follows from triangle inequality while the second follows from the definition of $\near$. 
Therefore, if we define the assignment cost vectors 
corresponding to $\phi$, $\phi'$, and $\phi^*$
as $\zf = \{d(v,\phi):v \in C\}$, $\zs = \{d(v,\phi'):v \in C\}$, and $\zv = \{d(v,\phi^*):v \in C\}$ respectively,
the above equation implies
    $\zs \leq 2\zf + \zv$.
Now note that the $\pnorm$ is a monotone norm on these vectors, and therefore,
\[
\pnorm(S;\phi') = \pnorm(\zs) \leq 2\pnorm(\zf) + \pnorm(\zv) = 2\pnorm(S^*;\phi^*) + \pnorm(S;\phi)
\]
The proof is complete by noting $\pnorm(S^*;\phi^*) = \optf(\cI)$ and $\pnorm(S;\phi)\leq \rho\cdot\optf(\cI)$.
%
\end{proof}



\subsection{Algorithm for the \fairCluAsgn problem}
\label{sec:assignment}

\def\LP{\mathsf{LP}}
\def\sx{x^{\star}}
\def\LPP{\mathsf{LP2}}

To complete the proof of \Cref{thm:fairclu}, we need to give an algorithm for the \fairCluAsgn problem. We present this in~\Cref{algo:faircluassgn}. The following theorem then establishes our main result.
\begin{restatable}{theorem}{ThmfairAssgn}
	\label{thm:fairAssgn}
	There exists a $(4\Delta+3)$-violating algorithm for the
	\fairCluAsgn problem.
\end{restatable}
\begin{proof}
Fix an instance $\cJ$ of the problem. 
We start by writing a natural LP-relaxation\footnote{This makes sense only for finite $p$. See~\Cref{rem:kcen}}.
\begin{subequations}
\label{eqn:violateLP} 
\begin{alignat}{4}
\LP := \min  &\sum\limits_{v \in C, f\in S}d(v,f)^p  x_{v,f} && \qquad x_{v,f} \in [0,1],~~ \forall v\in C, f\in S \label{eq:lp} \tag{LP}\\ 
      \beta_i \sum\limits_{v\in C} x_{v,f} ~~\leq~~ & \sum_{v\in C_i} x_{v,f}  ~~\leq~~ \alpha_i \sum\limits_{v\in C} x_{v,f} &&\qquad \forall f \in S, \forall i \in [\ell]   \label{eqn:violate:P1}   \\
   & \sum_{f \in S} x_{v,f}
    ~~=~~1 && \qquad\forall v \in C \label{eqn:violate:P3} 
\end{alignat}
\end{subequations}
\begin{claim}
	$\LP \leq \opta(\cJ)^p$. 
\end{claim}
\begin{proof}
Given an optimal solution $\phi^*$ of $\cJ$, set $x_{v,f} = 1$ iff $\phi^*(v) = f$. This trivially satisfies the fairness conditions.
Observe $\pnorm(S;\phi^*)^p$ is precisely the objective cost.
\end{proof}
Let $\sx$ be an optimum solution to the above LP. 
Note that $\sx$ could have many coordinates {\em fractional}. In~\Cref{algo:faircluassgn}, we {\em iteratively round} $\sx$ to an {\em integral} solution with the same or better value, but which violates the fairness constraints by at most $4\Delta+3$. Our algorithm effectively
simulates an algorithm for {\em minimum degree-bounded matroid basis} problem ($\mbdmb$ henceforth) due to Kir\'aly et al.~\cite{kiraly2012degree}.
In this problem one is given a matroid $M = (X,\cI)$, costs on elements in $X$, 
a hypergraph $H = (X,\cE)$, and functions $f:\cE\to \R$ and $g:\cE\to \R$ such that
$f(e) \le g(e)$ for all $e\in \cE$. The objective is to find the minimum cost basis $B\subseteq X$ such that for all $e\in \cE$, $f(e) \leq |B\cap e| \le g(e)$.  We state the main result in Kir\'aly et al~\cite{kiraly2012degree} below.
\begin{theorem}[Paraphrasing of Theorem 1 in ~\cite{kiraly2012degree}]
	\label{thm:matroid}
	There exists a polynomial time algorithm 
	that outputs a basis $B$ of cost at most $\opt$, such that
	$f(e) -2\Delta_H +1\leq |B \cap e|\leq g(e) +2\Delta_H -1 $ for each
	edge $e \in \cE$ of the hypergraph, where $\Delta_H = max_{v\in X} 
	|\{e \in E_H: v\in e\}|$ is the maximum degree of a vertex in
	the hypergraph $H$, and $\opt$ is the cost of the natural LP relaxation.
\end{theorem}

To complete the proof of the main theorem, we first construct an instance of the \mbdmb
problem using $\sx$. Then we appeal to~\Cref{thm:matroid} to argue about
the quality of our algorithm.

Let $E$ be the set of $(v,f)$ pairs with $\sx_{v,f} > 0$.
For a point $v\in C$, let $E_v$ denote the set of edges in $E$ incident on $v$. 
Define $\cF :=\{F\subseteq E : |F \cap E_v| \leq 1~~\forall v \in C\}$ to be collection of edges which ``hit'' every client at most once. The pair $M = (E,\cF)$ is a well known combinatorial object called a (partition) {\em matroid}.
For each element $(v,f)$ of this matroid $M$, we denotes its cost to be $c(v,f) := d(v,f)^p$.

Next we define a {\em hypergraph} $H = (E, \cE)$. 
For each $f\in S$ and $i\in [\ell]$, let  $E_{f,i} \subseteq E$ consisting of pairs $(v,f)\in E$ for $v\in C_i$.
Let $E_f := \cup_{i=1}^{\ell} E_{f,i}$. Each of these $E_{f,i}$'s and $E_f$'s are added to the collection of hyperedges $\cE$.
Next, let $T_f := \sum_{v\in C} x^{\star}_{v,f}$ be the total fractional assignment on $f$. Similarly, for all $i\in [\ell]$, define
$T_{f,i} :=  \sum_{v\in C_i} x^{\star}_{v,f}$. 
Note that, both $T_f$ and $T_{f,i}$ can be fractional. For every $e \in E_{f,i}$, we define $f(e) := \floor{T_{f,i}}$ and $g(e) = \ceil{T_{f,i}}$.
For each $e\in E_f$, we denote $f(e) = \floor{T_f}$ and $g(e) = \ceil{T_f}$. This completes the 
construction of the $\mbdmb$ instance.

Now we can apply~\Cref{thm:matroid} to obtain a basis $B$ of matroid $M$ with the properties mentioned. Note that for our hypergraph $\Delta_H \leq \Delta+1$ where $\Delta$ is the maximum number of groups a client can be in.
This is because every pair $(v,f)$ belongs to $E_{f}$ and $E_{f,i}$'s for all $C_i$'s containing $v$. 
Also note that any basis corresponds to an assignment $\phi:C\to S$ of all clients. Furthermore, the cost of the basis is precisely
$\pnorm(S;\phi)^p$. Since this cost is $\leq \LP \leq \optf(\cJ)^p$, we get that $\pnorm(S;\phi)\leq \optf(\cJ)$.
We now need to argue about the violation.

Fix a server
$f$ and a client group $C_i$. Let  $\overline{T}_f$
and  $\overline{T}_{f,i}$ denote the number of clients 
assigned to $f$ and the number of clients from $C_i$ that
are assigned to $f$ respectively (by the integral assignment).
Then, by~\Cref{thm:matroid}, $\floor{T_f} - 2\Delta - 1 \leq \overline{T}_f \leq \ceil{T_f} + 2\Delta + 1$
and $\floor{T_{f,i}} - 2\Delta - 1 \leq \overline{T}_{f,i} \leq \ceil{T_{f,i}} + 2\Delta + 1$ (using $\Delta_H \leq \Delta+1$).
Now consider~\cref{eqn:fairclu:majority}. 
Since, $T_{f,i} \leq \alpha_i T_{f}$ (as the LP solution is feasible), 
\[ 
\overline{T}_{f,i} \leq \ceil{\alpha_i T_f} + 2\Delta +1
\leq \alpha_i \floor{T_{f}} + 2\Delta + 2 
\leq \alpha_i (\overline{T}_f + 2\Delta + 1) + 2\Delta +2
\leq \alpha_i \overline{T}_f + (4\Delta + 3) \,,
\]
where the second and last inequality follows as $\alpha_i \leq 1$.
We can similarly argue about~\cref{eqn:fairclu:minority}.
This completes the proof of~\Cref{thm:fairAssgn}.\qedhere
\end{proof}
However, rather than constructing an
\mbdmb instance explicitly, we write a natural LP-relaxation more suitable
to the task --- this is given in~\cref{eqn:iterateLP}. For the sake
of completeness, we give the details of our algorithm in~\cref{algo:faircluassgn}.
\begin{algorithm}[!ht]
    \caption{Algorithm for the \fairCluAsgn problem}
    \label{algo:faircluassgn}
    \begin{algorithmic}[1]
        \Procedure{FairAssignment}{$(\cX,d)$, $S$,  $C=\cup_{i=1}^{\ell}C_i$, $\veca, \vecb \in [0,1]^{\ell}$}
            \State $\hat{\phi}(v) = \emptyset$ for all $v\in C$
            \State solve the $\LP$ given in~\cref{eqn:violateLP}, let $\sx$
                    be an optimal solution
            \State for each $x^{\star}_{v,f}=1$, set $\hat{\phi}(v)=f$ and remove $v$ from $C$ (and relevant $C_i$s).
            \State let $T_f := \sum_{v\in C} x^{\star}_{v,f}$ for all $f\in S$
            \State let $T_{f,i} :=  \sum_{v\in  C_i} x^{\star}_{v,f}$ for all $i\in [\ell]$ and $f\in S$
            \State construct $\LPP$ as given in~\cref{eqn:iterateLP},
                   only with variables $x_{v,f}$ such that $x^{\star}_{v,f}>0$
            \While {there exists a $v\in C$ such that $\hat{\phi}(v)=\emptyset$}
                \State solve $\LPP$, let $x^{\star}$ be an optimal solution
                \State for each $x^{\star}_{v,f}=0$, delete the variable $x^{\star}_{v,f}$ from $\LPP$
                \State for each $x^{\star}_{v,f}=1$, set $\hat{\phi}(v)=f$ and remove $v$ from $C$ (and relevant $C_i$s). Reduce $T_f$ and relevant $T_{f,i}$'s by 1.
                \State for every $i\in [\ell]$ and $f\in S$, if $|x^{\star}_{v,f} : 0<x^{\star}_{v,f} <1,~v\in C_i| \leq 2(\Delta+1)$ remove the respective constraint in~\cref{eqn:violate:iterate:P1}
                \State for every $f\in S$, if $|x^{\star}_{v,f} : 0<x^{\star}_{v,f} <1,~v \in C| \leq 2(\Delta+1)$ remove the respective constraint in~\cref{eqn:violate:iterate:P2}
            \EndWhile
        \EndProcedure
    \end{algorithmic}
\end{algorithm}
\begin{subequations}
\label{eqn:iterateLP} 
\begin{alignat}{4}
\LPP := \min & \sum\limits_{v \in C, f\in S}d(v,f)^px_{v,f} && \qquad x_{v,f} \in [0,1],~~ \forall v\in C, f\in S \label{eq:iterate:lp} \\ 
     & \floor{T_{f}} ~~\leq ~~ \sum_{v\in C} x_{v,f}  ~~\leq~~ \ceil{T_{f}} &&\qquad \forall f \in S, \forall i \in [\ell]   \label{eqn:violate:iterate:P2}  \\
     & \floor{T_{f,i}} ~~\leq ~~ \sum_{v\in C_i} x_{v,f}  ~~\leq~~ \ceil{T_{f,i}} &&\qquad \forall f \in S, \forall i \in [\ell]   \label{eqn:violate:iterate:P1}  \\
   & \sum_{f \in S} x_{v,f}
    ~~=~~1 && \qquad\forall v \in C \label{eqn:violate:iterate:P3} 
\end{alignat}
\end{subequations}
\begin{remark}\label{rem:kcen}
	For the case of $p = \infty$, the objective function of~\cref{eq:lp} doesn't make sense. Instead, one proceeds as follows.
	We begin with a guess $G$ of $\opta(\cJ)$; we set $x_{v,f} = 0$ for all pairs with $d(v,f) > G$. We then check if~\cref{eqn:violate:P1,eqn:violate:P3} 
	have a feasible solution. If they do not, then our guess $G$ is infeasible (too small). If they do, then the proof given above returns an assignment which violates~\cref{eqn:fairclu:majority,eqn:fairclu:minority} by additive $4\Delta+3$, and satisfies $d(v,\phi(v))\leq G$ for all $v\in C$. 
\end{remark}

\begin{remark}\label{rem:delta1}
	When $\Delta = 1$, that is, the $C_i$'s are disjoint, we can get an improved $+3$ additive violation (instead of $+7$). Instead of using~\Cref{thm:matroid}, 
	we  use the generalized assignment problem (GAP) rounding technique by Shmoys and Tardos~\cite{ShmoysT93} to achieve this. 
\end{remark}
\begin{remark}
Is having a bicriteria approximation necessary? We do not know. The nub is the \fairCluAsgn problem.
It is not hard to show that deciding whether a $\lambda$-violating solution exists with $\lambda = 0$ under the given definition is NP-hard.
\footnote{A simple reduction from the 3D-matching problem.}
However, an algorithm with $\lambda=0$ and cost within a constant factor of $\opta(\cJ)$ is not ruled out. This is an interesting open question.
\end{remark}

\section{Lower-bounded clustering}
\label{fn}
In this section we show a simple application of our technique which solves the lower bounded clustering problem.
The problem arises when, for example, one wants to ensure anonymity~\cite{AggarwalFKKPTZ06} and is called ``private clustering'' in~\cite{RS18}.
For the $p=\infty$ norm, that is the {\em lower bounded $k$-center} problem, there is a $3$-approximation known~\cite{AggarwalFKKPTZ06} for the problem.
For the $p=1$ norm, that is the {\em lower bounded $k$-median problem}, there are $O(1)$-approximation algorithms\footnote{Actually, the papers of~\cite{Svitkina10,AhmadianS12} consider the facility location version without any constraints on the number of facilities. A later paper by Ahmadian and Swamy~\cite{AhmadianS16} mentions that these algorithms imply $O(1)$-approximations for the $k$-median version. The constant is not specified.\label{fn:lb}}~\cite{Svitkina10,AhmadianS12} although the constants are large.
In contrast, we show simple algorithms with much better constants in $O(2^k \poly(n))$ time.

\begin{definition}[\lbclu]
\label{def:lbclu}
The input is a \clu instance, and an integer $L \in [|C|]$.
The objective is to open a set of facilities $S \in \cF$ with $|S| \leq k$, 
and find an assignment function 
$\phi: C \rightarrow{ S}$ of clients to the opened facilities
so that (a) $\pnorm(S;\phi)$ is minimized, and (b) for every $f\in S$, we have $|\{v\in C: \phi(v)=f\}| \geq L$.
\end{definition}
\ThmLB*
\noindent
Using the best known polynomial time algorithm for the $k$-median problem due to Byrka \etal~\cite{ByrkaPRS14} and best known FPT-algorithm due to Cohen-Addad \etal~\cite{CohenGKLL19}, we get the following corollary.
\begin{theorem}\label{thm:lbmed}
	There is a $4.676$-factor approximation algorithm for the lower bounded $k$-median running in $O(2^k \cdot \poly(n))$ time. 
	There is a $3.736$-factor approximation algorithm for lower bounded $k$-median running in time $k^{O(k)}\poly(n)$ time.
\end{theorem}
\begin{remark}
As in the case of fair clustering, \Cref{thm:lbfpt} holds even when there are more general constraints on the centers. 
Therefore, for instance, in $O(2^k\poly(n))$ time, we can get a $34$-approximation for the lower bounded knapsack median problem due to the knapsack median result~\cite{Swamy16},
and a $5$-approximation for the lower bounded center problem even when the total weight of the centers is at most a bound and the set of centers need be an independent set of a matroid~\cite{ChakrabartyN18}.
\end{remark}
\begin{proof}[Proof of~\Cref{thm:lbfpt}]
	The proof is nearly identical to that of~\Cref{fairclutoasgnthm}.
	Given an instance $\cI$ of the \lbclu problem, we first run algorithm $\cA$ to get $(S,\phi)$ with the property
	$\pnorm(S;\phi) \le \rho\cdot \optv(\cI) \le \rho\cdot \optl(\cI)$. 
	
	Now we construct $2^k$ instances of the $b$-matching problem with lower bounds. For \emph{every} subset $T\subseteq S$
	we construct a complete bipartite graph on $(C\cup T)$ with cost of edge $c(v,f) := d(v,f)^p$ for all $v\in C, f\in T$.
	There is a lower bound of $1$ on every $v\in C$ and $L$ on every $f\in T$. We find a minimum cost matching satisfying these.
	Given a matching $M$, the assignment is the natural one: $\phi(v) = f$ if $(v,f)\in M$. We return the assignment $\hat{\phi}$ of minimum cost
	among these $2^k$ possibilities. Note that $\pnorm(T;\hat{\phi})$ equals the $(1/p)$th power of the cost of this matching.
	
	To analyze the above algorithm, we need to show the existence of some subset $T\subseteq S$ such that the minimum cost lower bounded matching $M$ has 
	$c(M)^{1/p} \leq (\rho+2) \optl(\cI)$. The proof is very similar to the proof of~\Cref{lemma:lowcost}. 
Suppose $(S^*,\rho^*)$ is the optimal solution for $\cI$ of cost $\optl(\cI)$. 
%
For any $f^* \in S^*$, define $\near(f^*)$ to be its closest facility in $S$.
Let $T := \{f \in S: \exists f^* \in S^* \text{s.t. } f = \near(f^*)\}$; by definition $T \subseteq S$.
Define the matching where for each $v\in C$ we match $v$ to $\near(\phi^*(v)) \in T$. 
As in the proof of~\Cref{lemma:lowcost}, one can show that $d(v,\near(\phi^*(v))) \le 2d(v,\phi^*(v)) + d(v,\phi(v))$.
Thus, the $1/p$th power of the cost of the matching is at most $(\rho+2)\cdot\optl(\cI)$.
Furthermore, for every $f\in T$, the number of clients assigned to it is at least the number of clients assigned to an $f^*$ with $\near(f^*) = f$.
This is $\geq L$.
\end{proof}





\begin{remark}
	Finally, we mention that the above prove easily generalizes for the notion of {\em strong privacy} proposed by R\"osner and Schmidt~\cite{RS18}.
	In this, the client set is {\em partitioned} into groups $C_1,\ldots, C_\ell$, and the goal is to assign clients to open facilities such that
	for every facility the number of clients from a group $C_i$ is at least $L_i$. We can generalize~\Cref{thm:lbfpt} to get a $(\rho+2)$-approximation 
	algorithm for this problem running in $O(2^k\poly(n))$ time whenever there was a $\rho$-approximation possible for the vanilla clustering version.
%
%
\end{remark}

\section{Experiments}
\label{sec:experiments}
\graphicspath{./figures}

\def\bank{{\tt bank}~}
\def\census{{\tt census}~}
\def\creditcard{{\tt creditcard}~}
\def\cens1990{{\tt census1990}~}
In this section, we perform empirical evaluation
of our algorithm. Based on our experiments, we report five key findings:
\begin{itemize}
     \item Vanilla clustering algorithms are quite {\em unfair} even when measured against relaxed settings of $\alpha$ and $\beta$. In contrast,
     our algorithm's additive violation is almost always less than 3, even with $\Delta=2$, across a wide range of parameters
     (see~\cref{sec:fairness_vs_vanilla}).
     \item The cost of our fair clustering is at most 15\% more than (unfair) vanilla cost for $k \leq 10$ as in \cref{fig:cost_vs_van}. In fact, we see (in ~\cref{fig:cost_vs_almostfair}) that our algorithm's cost is \emph{very close} to the {\em absolute best} fair clustering that allows additive violations!
    Furthermore, our results for $k$-median significantly improve over the costs reported in Chierichetti \etal\cite{CKLV18} and Backurs \etal\cite{Backurs2019} (see~\Cref{table:cost_vs_cklv}
).
     \item For the case of overlapping protected groups ($\Delta > 1$), enforcing fairness with respect to one sensitive attribute (say gender) can lead to 
     unfairness with respect to another (say race). This empirical evidence stresses the importance of considering $\Delta > 1$ (see~\cref{fig:additive_fig} in~\Cref{subsec:appendix:overlap}). 
     \item In~\cref{sec:runtime}, we provide experiments to gauge the running time of our algorithm in practice on large datasets. Even though the focus of this paper is not optimizing the running time, we observe that our algorithm for the k-means objective finds a fair solution for the \cens1990 dataset with 500K points and 13 features in less than 30 minutes (see~\Cref{table:timing}).
     \item Finally, we study how the cost of our fair clustering algorithm changes with the strictness of the fairness conditions. This enables the user to figure out the trade-offs between fairness and utility and make an informed decision about which threshold to choose (see~\Cref{sec:appendix_tuning}). 
 \end{itemize}


\textbf{Settings.} We implement our algorithm in Python 3.6 and run
all our experiments on a Macbook Air with a 1.8 GHz Intel Core i5 Processor and 8 GB 1600 MHz DDR3 memory. We use CPLEX\cite{CPLEX} for solving LP's. Our codes are available on GitHub\footnote{\href{https://github.com/nicolasjulioflores/fair_algorithms_for_clustering}{https://github.com/nicolasjulioflores/fair\textunderscore algorithms\textunderscore for\textunderscore clustering}} for public use.

\textbf{Datasets.} We use four datasets from the UCI repository~\cite{Dua:2019}:
\footnote{\href{https://archive.ics.uci.edu/ml/datasets/}{https://archive.ics.uci.edu/ml/datasets/}}
\begin{inparaenum}[\bfseries (1)]
\item \bank~\cite{bank-paper} with 4,521 points, corresponding to phone calls from a marketing campaign by a
Portuguese banking institution.
\item \census~\cite{census-paper} with 32,561 points, representing information about individuals extracted from the 1994 US \census.
\item \creditcard~\cite{creditcard-paper} with 30,000 points, related to information on credit card holders from a certain credit card in Taiwan.
\item \cens1990~\cite{meek2002learningcurve} with 2,458,285 points, taken from the 1990 US Census, which we use for run time analysis. 
\end{inparaenum}
For each of the datasets, we select a set of numerical attributes to represent the records in the Euclidean space. We also choose two sensitive 
attributes for each dataset (e.g. sex and race for \census) and create protected groups based on their values.  ~\Cref{table:data} contains a more detailed description of the datasets and our features.

\begin{table}[!ht]
\centering
\caption{For each dataset, the coordinates are the numeric attributes used to determined the position of each record in the Euclidean space. The sensitive attributes determines protected groups.}
\begin{tabular}{ llrl }
\rule{0pt}{4ex}
\textbf{Dataset}& \textbf{Coordinates} & \textbf{\shortstack{Sensitive\\  attributes}}&\textbf{Protected groups} \\ 
 \hline 
 \rule{0pt}{2.5ex}
 \multirow{2}*\textbf{\bank}
 & age, balance, duration & marital & married, single, divorced \\
 \cmidrule(r){3-4}
 &  & default &  yes, no \\
 \hline 
 \rule{0pt}{2.5ex}
 \multirow{3}*\textbf{\census}
 & age, education-num,  & sex & female, male\\
 \cmidrule(r){3-4}
  &final-weight, capital-gain, & race & Amer-ind, asian-pac-isl,\\
  \rule{0pt}{2.5ex}
   & hours-per-week && black, other, white\\
   \hline
   \rule{0pt}{2.5ex}
  \multirow{2}*\textbf{\creditcard} & age, bill-amt 1 --- 6, & marriage & married, single, other, null \\
  \cmidrule(r){3-4}
    & limit-bal, pay-amt 1 --- 6  & education& 7 groups\\
    \hline
    \rule{0pt}{2.5ex}
    \multirow{4}*\textbf{\cens1990}
    & dAncstry1, dAncstry2, iAvail, & dAge & $8$ groups \\
    \cmidrule(r){3-4}
    & iCitizen, iClass, dDepart, iFertil,  & iSex & female, male\\
    \rule{0pt}{2.5ex}
    & iDisabl1, iDisabl2, iEnglish,  \\
    \rule{0pt}{2.5ex}
    & iFeb55, dHispanic, dHour89 \\
    \hline
\end{tabular}
\label{table:data}
\end{table}

\textbf{Measurements.} 
For any clustering, we mainly focus on two metrics. One is the \emph{cost of fairness}, that is, the ratio of the objective values of the fair clustering over the vanilla clustering. The other is \emph{balance}, the measure of unfairness. To define balance, we generalize the notion found by Chierichetti \etal~\cite{CKLV18},
We define two intermediate values $r_i$, the representation of group $i$ in the dataset and $r_i(f)$, the representation of group $i$ in cluster $f$ as
$r_i := {| C_i |}/{| C |}$ and
$r_i(f) := {| C_i(f) |}/{| C(f) |}$.
Using these two values, balance is defined as
$
\balance(f) := \min\{r_i/r_i(f) ,r_i(f)/r_i\}~~ \forall i \in [\ell]
$.
Although in theory the values of $\alpha$, $\beta$ for a given group $i$ can be set arbitrarily, in practice they are best set with respect to $r_i$, the ratio of the group in the dataset. Furthermore, to reduce the degrees of freedom, we parameterize $\beta$ and $\alpha$ by a single variable $\delta$ such that $\beta_i = r_i(1 - \delta)$ and $\alpha_i = r_i/(1 - \delta)$. Thus, we can interpret $\delta$ as how loose our fairness condition is. This is because $\delta = 0$ corresponds to each group in each cluster having exactly the same ratio as that group in the dataset, and  $\delta = 1$ corresponds to no fairness constraints at all. For all of the experiments, 
we set $\delta=0.2$ (corresponding to the common interpretation of the $80\%$-rule of DI doctrine), and
use $\Delta=2$, unless otherwise specified.

\textbf{Algorithms.} 
For vanilla $k$-center, we use a $2$-approx. algorithm due to Gonzalez~\cite{Gonzalez85}.
For vanilla $k$-median, we use the single-swap $5$-approx. algorithm by Arya \etal~\cite{AryaGKMMP04}, augment it with the $D$-sampling procedure by~\cite{AV07} for initial center section, and take the best out of $5$ trials. For $k$-means, we use the $k$-means++ implementation of~\cite{scikit-learn}. 

\subsection{Fairness comparison with vanilla clustering}\label{sec:fairness_vs_vanilla}
In~\cref{fig:vanilla_fig} we motivate our discussion of fairness 
by demonstrating the unfairness of vanilla clustering and fairness of our algorithm. 
On the $x$-axis, we compare three solutions: (1) our algorithm (labelled ``ALG''),
(2) fractional solution to the \fairCluAsgn LP in ~\Cref{eqn:violateLP} (labelled
``Partial''), and (3) vanilla $k$-means (labelled ``VC'').
Below these labels, we record the \emph{cost of fairness}.
We set $\delta = 0.2$ 
and $k=4$. Along the $y$ axis, we plot the balance metric defined above 
for the three largest clusters for each of these clustering. The dotted line at $0.8$ is the goal balance for $\delta = 0.2$. The lowest balance for any cluster for our algorithm is $0.75$ (for \census), whereas vanilla can be as bad as $0$ (for \bank); ``partial'' is, of course, always fair (at least $0.8$).
\begin{figure}[!ht]
\centering
\includegraphics[width=0.9\columnwidth]{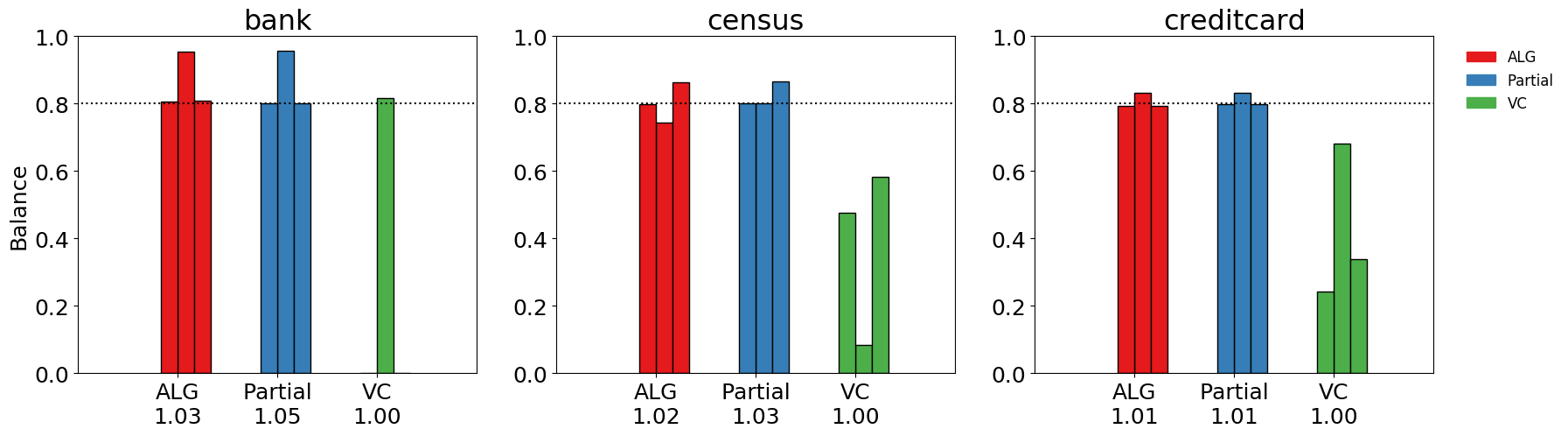}
\caption{Comparison of our algorithm (ALG) versus vanilla clustering (VC) in terms of $\balance$ for the $k$-means objective. }
\label{fig:vanilla_fig}
\end{figure}
We observe that the maximum additive violation of our algorithm is only $3$ (much better than our theoretical bound of $4\Delta+3)$),
for a large range of values of $\delta$ and $k$,
whereas vanilla $k$-means can be unfair by quite a large margin.
(see~\cref{fig:additive_fig} below and~\Cref{table:additive}).
\begin{figure}[!ht]
\includegraphics[width=0.9\columnwidth]{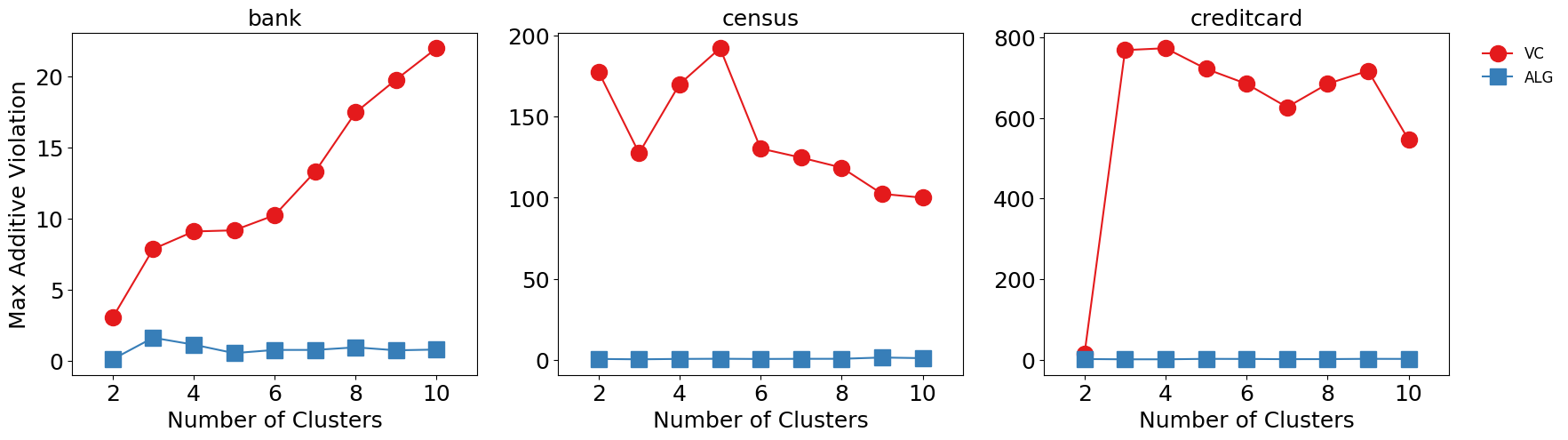}
\centering
\caption{Comparison of  the maximum additive violation (for $\delta = 0.2$ and $\Delta = 2$) over all clusters and all groups between our algorithm (ALG) and vanilla (VC), using the $k$-means objective.}
\label{fig:additive_fig}
\end{figure}

\begin{table}[!ht]
\caption{The maximum additive violation across a range of $\delta$ of our algorithm compared to vanilla $k$-means. For each $\delta$, we take maximum over $k$, for $k \in [2,10]$ on all datasets.}
\centering
\begin{tabular}{  l c c c c c c c  c }

\textbf{$\delta$} \Tstrut & 0.01 & 0.05 & 0.1 & 0.2 & 0.3 & 0.4 & 0.5 & Vanilla ($\delta = 0.2$) \\ [0.5ex]
 \hline
\textbf{\bank} \Tstrut & 1.45 & 1.17 & 1.39 & \textbf{1.54} & 1.19 & 1.15 & 1.03 & \textbf{21.99} \\ [1ex]
\textbf{\census} & 1.44 & 1.53 & \textbf{1.89} & 1.08 & 1.18 & 0.97 & 1.03 & \textbf{773.19} \\  [1ex]
\textbf{\creditcard} & \textbf{3.02} & 2.32 & 2.11 & 2.29 & 2.03 & 1.63 & 1.03 & \textbf{192.01} \\ [1ex]
 \hline
\end{tabular}
\label{table:additive}
\end{table}

\subsection{Cost analysis} We evaluate the cost of our algorithm for $k$-means objective with respect to the vanilla clustering cost. \Cref{fig:cost_vs_van} shows that the cost of our algorithm for $k \leq 10$ is at most 15\% more than the vanilla cost for all datasets. Interestingly, for \creditcard, even though the vanilla solution is extremely unfair as demonstrated earlier, cost of fairness is at most 6\% which indicates that the vanilla centers are in the ``right place''.
\begin{figure}[!ht]
  \centering
    \includegraphics[width=0.3\columnwidth]{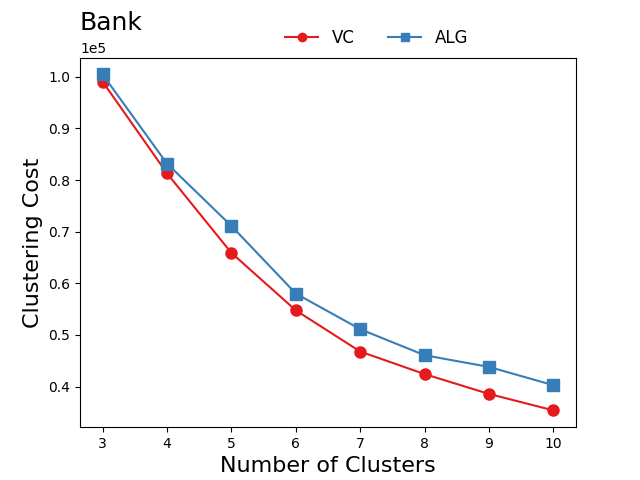}
    \includegraphics[width=0.3\columnwidth]{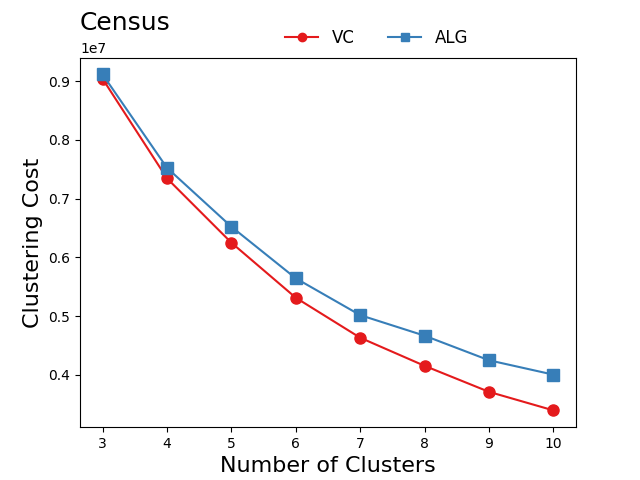}
    \includegraphics[width=0.3\columnwidth]{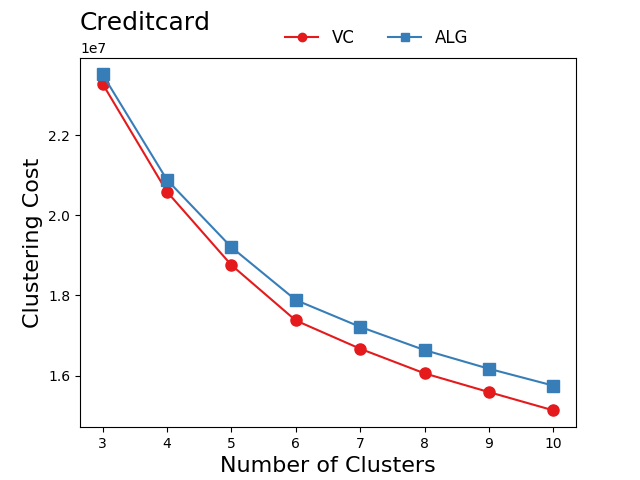}
  \caption{Our algorithm's cost (ALG) versus the vanilla clustering cost (VC) for $k$-means objective.}
  \label{fig:cost_vs_van}
\end{figure}

Our results in Table \ref{table:cost_vs_cklv} confirm that we outperform \cite{CKLV18} and \cite{Backurs2019} in terms of cost. To match \cite{CKLV18} and \cite{Backurs2019}, we sub-sample \bank and \census to 1000 and 600 respectively, declared only one sensitive attribute for each (i.e. marital for \bank and sex for \census), and tune the fairness parameters to enforce a balance of $0.5$. The data in \cref{table:cost_vs_cklv} is from \cite{Backurs2019} which is only available for $k$-median and $k = 20$.

\begin{table}[!ht]
\centering
\caption{Comparison of our clustering cost with \cite{Backurs2019} and \cite{CKLV18} for $k$-median ($k=20$).}
\begin{tabular}{ llll }
\rule{0pt}{3ex}
Dataset& \textbf{Ours} & Backurs \etal \cite{Backurs2019} & Chierichetti \etal\cite{CKLV18}  \\ [0.5ex]
 \hline
 \rule{0pt}{2ex}
 \textbf{\bank} & \boldmath$2.43 \times 10^5$ & $6.03 \times 10^5$ & $5.55 \times 10^5$\\
 \textbf{\census} & \boldmath$4.24\times 10^6$ & $24.1 \times 10^6$ & $36.5 \times 10^6$\\
 \hline
\end{tabular}
\label{table:cost_vs_cklv}
\end{table}

\def\LPFAIR{\mathsf{LP3}}
Next, we evaluate the cost of our algorithm for $k$-means objective with respect to the vanilla clustering cost and the \emph{almost fair LP} cost. The almost fair LP (\cref{eqn:almostfairLP}) is an LP relaxation of \fairClu, with variables for choosing the centers, except that we allow for a $\lambda$ additive violation in fairness. The cost of this LP is a lower-bound on the cost of \emph{any} fair clustering that violates fairness by at most an additive factor of $\lambda$.

\begin{subequations}
\label{eqn:almostfairLP} 
\begin{alignat}{6}
\LPFAIR := \min & \sum\limits_{v \in C, f\in S}d(v,f)^px_{v,f} && \qquad x_{v,f} \in [0,1],~~ \forall v\in C, f\in S \label{eq:fairlp}\\ 
    & \sum_{f \in S} x_{v,f}
    ~~=~~1 && \qquad\forall v \in C \label{eqn:fairlp:P3}  \\
   & x_{v,f}
    ~~\leq~~y_f && \qquad\forall v \in C, f\in S \label{eqn:fairlp:P4} \\ 
   & \sum_{f \in S} y_f
    ~~\leq~~k && \qquad \label{eqn:fairlp:P5} \\
     & \sum_{v\in C_i} x_{v,f}  
     ~~\leq~~ \alpha_i \sum\limits_{v\in C} x_{v,f} + \lambda &&\qquad \forall f \in S, \forall i \in [\ell]   \label{eqn:fairlp:violate:P1}  \\
    &\sum_{v\in C_i} x_{v,f} 
    ~~\geq~~ \beta_i \sum\limits_{v\in C} x_{v,f} - \lambda   && \qquad\forall f \in S, \forall i \in [\ell]   \label{eqn:fairlp:violate:P2} 
\end{alignat}
\end{subequations}

In \cref{fig:cost_vs_almostfair} we compare the cost of our algorithm with a lower-bound on the absolute best cost of any clustering that has the same amount of violation as ours. To be more precise, for any dataset we set $\lambda$ according to the maximum violation of our algorithm reported in \cref{table:additive} for $\delta = 0.2$ (e.g. $\lambda$ is $1.54$ for \bank, $1.08$ for \census, and $2.29$ for \creditcard). Then, we solve the almost fair LP for that $\lambda$ and compare its cost with our algorithm's cost over that dataset.

Since solving the almost fair LP on the whole data is infeasible (in terms of running time), we sub-sample \bank, \census, and \creditcard to 1000, 600 and 600 points respectively, and report the average costs over 10 trials. Also, we only consider one sensitive attribute, namely marital for \bank, sex for \census and education for \creditcard to further simplify the LP and decrease the running time. \cref{fig:cost_vs_van} shows that the cost of our algorithm is very close to the almost fair LP cost (at most 15\% more). Note that, since the cost of almost fair LP is a lower bound on the cost of \fairClu problem, we conclude that our cost is at most 15\% more than optimum in practice, which is much better than the proved $(\rho + 2)$ factor in \cref{thm:fairclu}.

\begin{figure}[!ht]
  \centering
  \includegraphics[width=0.3\columnwidth]{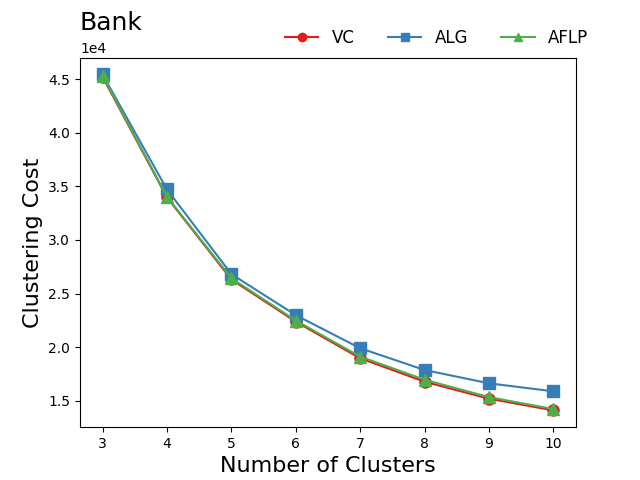}
  \includegraphics[width=0.3\columnwidth]{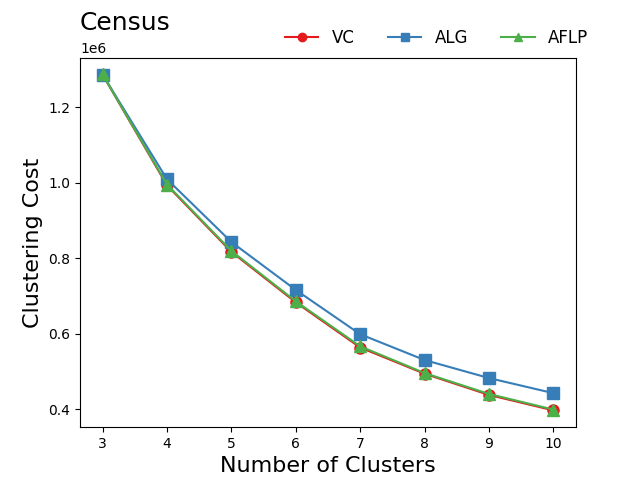}
    \includegraphics[width=0.3\columnwidth]{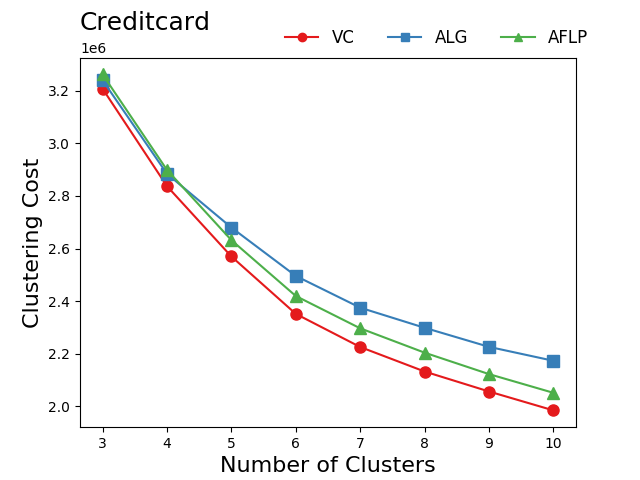}
  \caption{Average costs of vanilla clustering (VC), our algorithm (ALG), and almost fair LP (AFLP), for $k$-means objective, as a function of $k$.}
  \label{fig:cost_vs_almostfair}
\end{figure}

\subsection{The case of $\Delta > 1$}
\label{subsec:appendix:overlap}
In this section, we demonstrate the importance of considering $\Delta>1$ by
showing that enforcing fairness with respect to one
attribute (say gender) may lead to significant unfairness with respect to
another attribute (say race). In Figure \ref{fig:cf_fig}, we have
two plots for each dataset. 
In each plot, we compare three clustering: (1) Our algorithm 
with $\Delta=2$ (labelled ``both''); (2) and (3) Our algorithm with 
$\Delta=1$ with protected groups defined by the attribute on $x$-axis label. We set $\delta=0.2$ and $k=4$.  The clustering objective is $k$-means. 
Along $y$-axis, we measure the 
\emph{balance} metric for the three largest clusters for each of these
clustering. In each plot we only measure the {\em balance} for the 
attribute written in bold in the top right corner.

In datasets, such as \bank, we see that fairness with respect to only the marital attribute leads to a large amount of unfairness in the default attribute.  The fairest solution along both attributes is when they are both considered by our algorithm ($\Delta = 2$). Interestingly, there are datasets where fairness by one attribute is all that is needed. On the \census dataset, fairness by race leads to a fair solution on sex, but fairness
by sex leads to large amount of unfairness in race.

Finally, our results strongly suggest that finding a fair solution for two attributes is often only slightly more expensive (in terms of the clustering objective) than finding a fair solution for only one attribute.

\begin{figure}[!ht]
\includegraphics[width=\columnwidth]{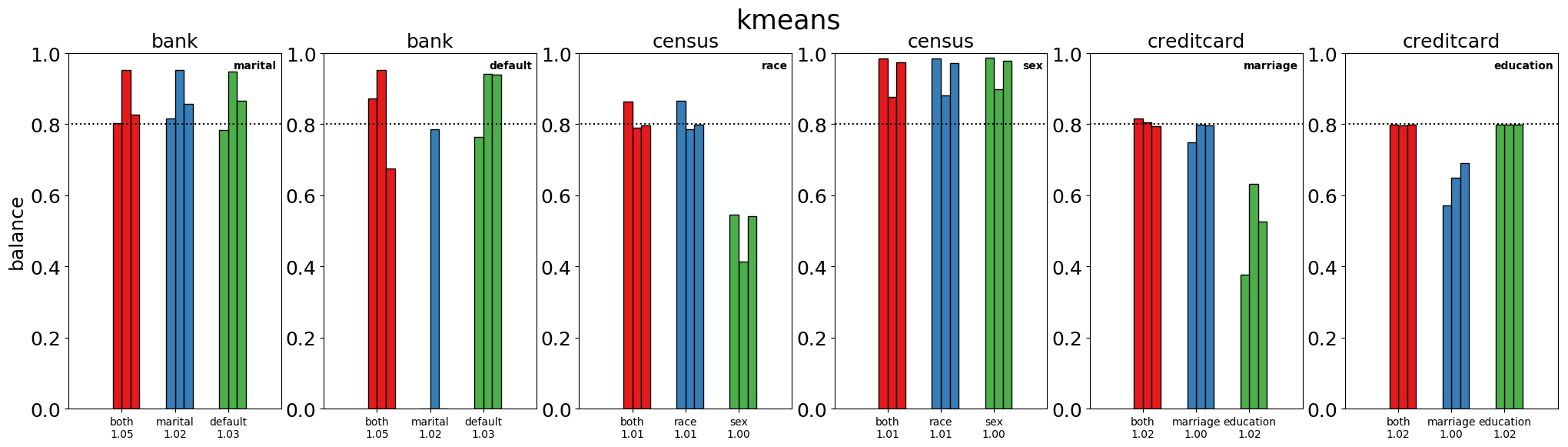}
\centering
\caption{Importance of considering $\Delta > 1$. Below these x labels is the \emph{cost of fairness} ratio. We report the \emph{balance} for the three largest clusters and include the dotted line at $0.8$ because we use $\delta = 0.2$.}
\label{fig:cf_fig}
\end{figure}

\subsection{Run time analysis}\label{sec:runtime}
In this paper, we focus on providing a framework and do not emphasize on 
run time optimization. Nevertheless, we note that our algorithm for the k-means objective finds a fair solution for the \cens1990 dataset with 500K points and 13 features in less than 30 minutes (see~\Cref{table:timing}). 
Even though our approach is based on iterative rounding method, 
in practice CPLEX solution to $\LP$ (\cref{eqn:violateLP})
is more than 99\% integral for each of our experiments. 
Hence, we never have to solve more than two or three LP. Also 
the number of variables in subsequent LPs are significantly small.
In contrast, if we attempt to frame $\LP$ (\cref{eqn:violateLP}) as an integer program instead, the CPLEX solver fails to find a solution in under an hour even with 40K points.


\begin{table}[!ht]
\caption{Runtime of our algorithm on subsampled data from \cens1990 for $k$-means ($k=3$).}
\centering
\begin{tabular}{  r  lllllll  }
\\\hline
 \rule{0pt}{2.5ex}
 \textbf{Number of sampled points} & 10K & 50K & 100K & 200K & 300K & 400K & 500K \\ 
 \textbf{Time (sec)} & 4.04 & 33.35 & 91.15 & 248.11 & 714.73 & 1202.89 & 1776.51 \Bstrut\\
 \hline
\end{tabular}
\label{table:timing}
\end{table}

\subsection{Tuning the fairness parameters}\label{sec:appendix_tuning}
In Figure \ref{fig:delta_fig}, we demonstrate the ability to tune the strictness of the fairness criteria by manipulating the parameter $\delta$. As $\delta$ approaches $1$, the ratio between the fair objective and original vanilla objective decreases to $1$. This suggests that the fair solution has recapitulated the vanilla clustering because our bounds are lax enough to do so.

\begin{figure}[h!]
\includegraphics[width =\columnwidth]{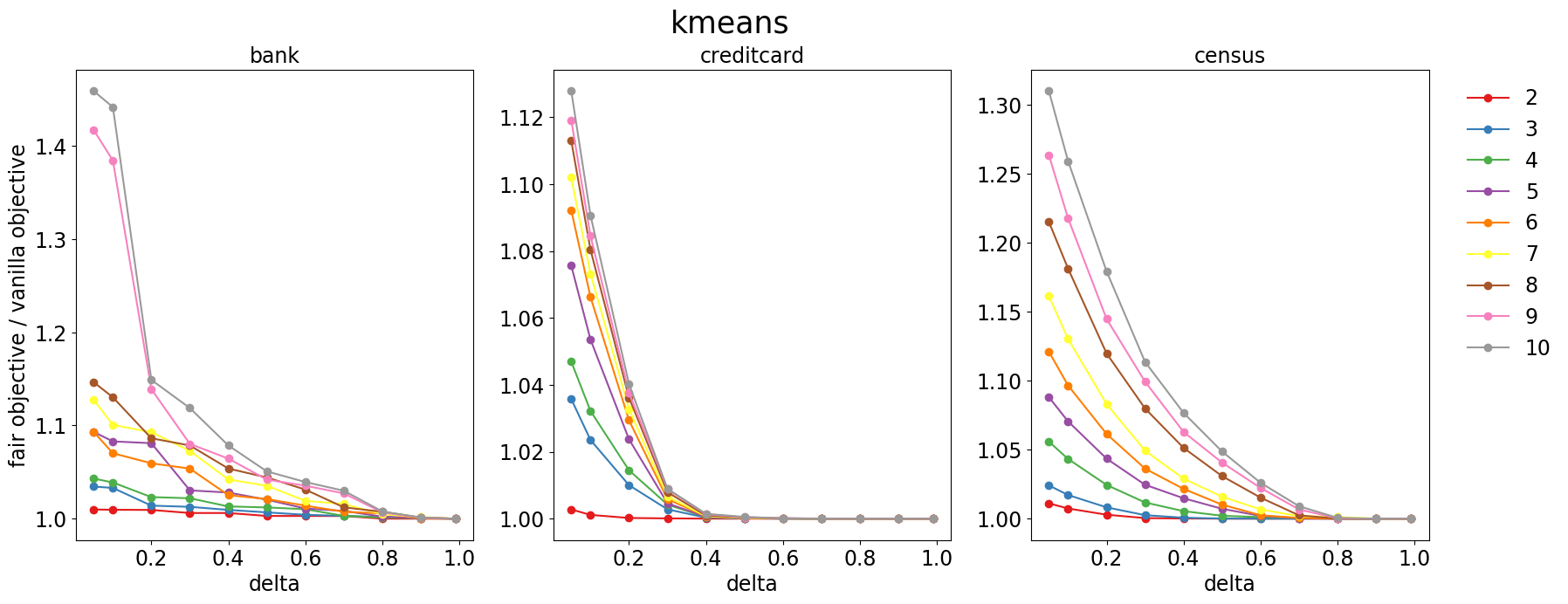}
\centering
\caption{We show the effects of varying $\delta$ (x-axis) on our algorithm's fair objective cost over the vanilla cost (y-axis).}
\label{fig:delta_fig}
\end{figure}

%
%
%
\newpage
 \bibliographystyle{plain}
 \bibliography{refs}

\begin{thebibliography}{10}

\bibitem{401}
{Supreme Court of the United States. Griggs v. Duke Power Co. 401 U.S. 424},
  March 8 1971.

\bibitem{Adamczyk18}
Marek Adamczyk, Jaroslaw Byrka, Jan Marcinkowski, Syed~M. Meesum, and Michal
  Wlodarczyk.
\newblock Constant factor {FPT} approximation for capacitated k-median.
\newblock {\em CoRR}, abs/1809.05791, 2018.

\bibitem{aggarwal2013data}
Charu~C Aggarwal and Chandan~K Reddy.
\newblock {\em Data clustering: algorithms and applications}.
\newblock CRC press, 2013.

\bibitem{AggarwalFKKPTZ06}
Gagan Aggarwal, Tom{\'a}s Feder, Krishnaram Kenthapadi, Samir Khuller, Rina
  Panigrahy, Dilys Thomas, and An~Zhu.
\newblock Achieving anonymity via clustering.
\newblock In {\em ACM Symposium on Principles of Database Systems}, 2006.

\bibitem{AC05}
Herman Aguinis and Wayne~F. Cascio.
\newblock {\em Applied Psychology in Human Resource Management (6th Edition)}.
\newblock Prentice Hall, 2005.

\bibitem{sara2019clustering}
Sara Ahmadian, Alessandro Epasto, Ravi Kumar, and Mohammad Mahdian.
\newblock Clustering without over-representation.
\newblock In {\em Proc. 25th Annual SIGKDD International Conference on
  Knowledge Discovery and Data Mining}, 2019.

\bibitem{AhmadianNSW17}
Sara Ahmadian, Ashkan Norouzi-Fard, Ola Svensson, and Justin Ward.
\newblock Better guarantees for $k$-means and {E}uclidean $k$-median by
  primal-dual algorithms.
\newblock In {\em Annual IEEE Symposium on Foundations of Computer Science},
  2017.

\bibitem{AhmadianS12}
Sara Ahmadian and Chaitanya Swamy.
\newblock Improved approximation guarantees for lower-bounded facility
  location.
\newblock In {\em Proceedings, Workshop on Approximation and Online Algorithms
  (WAOA)}, 2012.

\bibitem{AhmadianS16}
Sara Ahmadian and Chaitanya Swamy.
\newblock {Approximation Algorithms for Clustering Problems with Lower Bounds
  and Outliers}.
\newblock In {\em International Colloquium on Automata, Languages and
  Programming}, 2016.

\bibitem{propublica}
J.~Angwin, J.~Larson, S.~Mattu, and L.~Kirchner.
\newblock Machine bias: There’s software used across the country to predict
  future criminals. and it’s biased against blacks.
\newblock {\em ProPublica}, May 23 2016.

\bibitem{AV07}
David Arthur and Sergei Vassilvitskii.
\newblock K-means++: The advantages of careful seeding.
\newblock In {\em Proceedings of the Eighteenth Annual ACM-SIAM Symposium on
  Discrete Algorithms}, 2007.

\bibitem{Backurs2019}
Krzysztof Onak Baruch Schieber Ali~Vakilian Arturs~Backurs, Piotr~Indyk and Tal
  Wagner.
\newblock Scalable fair clustering.
\newblock In {\em Proc. 36th Proceedings, International Conference on Machine
  Learning (ICML)}, June 2019.

\bibitem{AryaGKMMP04}
Vijay Arya, Naveen Garg, Rohit Khandekar, Adam Meyerson, Kamesh Munagala, and
  Vinayaka Pandit.
\newblock Local search heuristics for k-median and facility location problems.
\newblock {\em SIAM Journal on computing}, 33(3):544--562, 2004.

\bibitem{Bercea2018}
Ioana~O. Bercea, Martin Gro{\ss}, Samir Khuller, Aounon Kumar, Clemens
  R{\"{o}}sner, Daniel~R. Schmidt, and Melanie Schmidt.
\newblock On the cost of essentially fair clusterings.
\newblock {\em CoRR}, abs/1811.10319, 2018.

\bibitem{BJK18}
Anup Bhattacharya, Ragesh Jaiswal, and Amit Kumar.
\newblock Faster algorithms for the constrained k-means problem.
\newblock {\em Theory of Computing Systems}, 62, 2018.

\bibitem{ByrkaPRS14}
Jaros{\l}aw Byrka, Thomas Pensyl, Bartosz Rybicki, Aravind Srinivasan, and Khoa
  Trinh.
\newblock An improved approximation for $k$-median, and positive correlation in
  budgeted optimization.
\newblock In {\em Annual ACM-SIAM Symposium on Discrete Algorithms}, 2014.

\bibitem{calders2010three}
Toon Calders and Sicco Verwer.
\newblock Three naive bayes approaches for discrimination-free classification.
\newblock {\em Data Mining and Knowledge Discovery}, 21(2):277--292, 2010.

\bibitem{Celis:2019}
L.~Elisa Celis, Lingxiao Huang, Vijay Keswani, and Nisheeth~K. Vishnoi.
\newblock Classification with fairness constraints: A meta-algorithm with
  provable guarantees.
\newblock In {\em Proceedings of the Conference on Fairness, Accountability,
  and Transparency}, FAT* '19, 2019.

\bibitem{celis2018multiwinner}
L~Elisa Celis, Lingxiao Huang, and Nisheeth~K Vishnoi.
\newblock Multiwinner voting with fairness constraints.
\newblock In {\em IJCAI}, pages 144--151, 2018.

\bibitem{celis2018}
L.~Elisa Celis, Damian Straszak, and Nisheeth~K. Vishnoi.
\newblock {Ranking with Fairness Constraints}.
\newblock In {\em Proc. 45th International Colloquium on Automata, Languages
  and Programming}, pages 28:1--28:15, 2018.

\bibitem{ChakrabartyN18}
Deeparnab Chakrabarty and Maryam Negahbani.
\newblock {Generalized Center Problems with Outliers}.
\newblock In {\em International Colloquium on Automata, Languages and
  Programming}, 2018.

\bibitem{CharikarG99}
Moses Charikar and Sudipto Guha.
\newblock Improved combinatorial algorithms for the facility location and
  $k$-median problems.
\newblock In {\em Annual IEEE Symposium on Foundations of Computer Science},
  1999.

\bibitem{CharikarGST02}
Moses Charikar, Sudipto Guha, {\'E}va Tardos, and David~B. Shmoys.
\newblock A constant-factor approximation algorithm for the $k$-median problem.
\newblock {\em J. Comput. Syst. Sci.}, 65(1):129--149, 2002.

\bibitem{chen2019proportionally}
Xingyu Chen, Brandon Fain, Charles Lyu, and Kamesh Munagala.
\newblock Proportionally fair clustering.
\newblock In {\em Proc. 36th Proceedings, International Conference on Machine
  Learning (ICML)}, June 2019.

\bibitem{CKLV18}
Flavio Chierichetti, Ravi Kumar, Silvio Lattanzi, and Sergei Vassilvitskii.
\newblock Fair clustering through fairlets.
\newblock In {\em Proc. 31st Conference on Neural Information Processing
  Systems}, pages 5029--5037, 2017.

\bibitem{chierichetti19a}
Flavio Chierichetti, Ravi Kumar, Silvio Lattanzi, and Sergei Vassilvtiskii.
\newblock Matroids, matchings, and fairness.
\newblock In {\em Proceedings of Machine Learning Research}, volume~89, 2019.

\bibitem{chouldechova2017fair}
Alexandra Chouldechova.
\newblock Fair prediction with disparate impact: A study of bias in recidivism
  prediction instruments.
\newblock {\em Big data}, 5(2):153--163, 2017.

\bibitem{CohenGKLL19}
Vincent Cohen-Addad, Anupam Gupta, Amit Kumar, Euiwoong Lee, and Jason Li.
\newblock Tight fpt approximations for $ k $-median and $ k $-means.
\newblock {\em International Colloquium on Automata, Languages and
  Programming}, 2019.

\bibitem{corbett2017algorithmic}
Sam Corbett-Davies, Emma Pierson, Avi Feller, Sharad Goel, and Aziz Huq.
\newblock Algorithmic decision making and the cost of fairness.
\newblock In {\em Proc. 23rd Annual SIGKDD International Conference on
  Knowledge Discovery and Data Mining}, pages 797--806. ACM, 2017.

\bibitem{DX15}
Hu~Ding and Jinhui Xu.
\newblock A unified framework for clustering constrained data without locality
  property.
\newblock In {\em Annual ACM-SIAM Symposium on Discrete Algorithms}, 2015.

\bibitem{googlead}
Jerry Dischler.
\newblock {Putting machine learning into the hands of every advertiser}.
\newblock
  \url{https://www.blog.google/technology/ads/machine-learning-hands-advertisers/},
  July 10 2018.

\bibitem{dressel2018accuracy}
Julia Dressel and Hany Farid.
\newblock The accuracy, fairness, and limits of predicting recidivism.
\newblock {\em Science advances}, 4, 2018.

\bibitem{Dua:2019}
Dheeru Dua and Casey Graff.
\newblock {UCI} machine learning repository, 2017.

\bibitem{dwork2012fairness}
Cynthia Dwork, Moritz Hardt, Toniann Pitassi, Omer Reingold, and Richard Zemel.
\newblock Fairness through awareness.
\newblock In {\em Proc. 3rd Conference on Innovations in Theoretical Computer
  Science}, pages 214--226. ACM, 2012.

\bibitem{EEOC}
The~U.S. EEOC.
\newblock Uniform guidelines on employee selection procedures, March 2 1979.

\bibitem{feldman2015certifying}
Michael Feldman, Sorelle~A Friedler, John Moeller, Carlos Scheidegger, and
  Suresh Venkatasubramanian.
\newblock Certifying and removing disparate impact.
\newblock In {\em Proc. 21st Annual SIGKDD International Conference on
  Knowledge Discovery and Data Mining}, pages 259--268, 2015.

\bibitem{FriedlerSV16}
Sorelle~A. Friedler, Carlos Scheidegger, and Suresh Venkatasubramanian.
\newblock On the (im)possibility of fairness.
\newblock {\em CoRR}, abs/1609.07236, 2016.

\bibitem{Gonzalez85}
Teofilo~F. Gonzalez.
\newblock {Clustering to Minimize the Maximum Intercluster Distance}.
\newblock {\em Theor. Comput. Sci.}, 38:293 -- 306, 1985.

\bibitem{GuptaT08}
Anupam Gupta and Kanat Tangwongsan.
\newblock Simpler analyses of local search algorithms for facility location.
\newblock {\em arXiv preprint arXiv:0809.2554}, 2008.

\bibitem{HochbaumS85a}
Dorit~S. Hochbaum and David~B. Shmoys.
\newblock A best possible heuristic for the $k$-center problem.
\newblock {\em Math.\ Oper.\ Res.}, 10(2):180--184, 1985.

\bibitem{hsu1979easy}
Wen-Lian Hsu and George~L Nemhauser.
\newblock Easy and hard bottleneck location problems.
\newblock {\em Discrete Applied Mathematics}, 1(3):209--215, 1979.

\bibitem{creditcard-paper}
Che-hui~Lien I-Cheng~Yeh.
\newblock The comparisons of data mining techniques for the predictive accuracy
  of probability of default of credit card clients.
\newblock {\em Expert Systems with Applications}, 2009.

\bibitem{CPLEX}
IBM.
\newblock Ibm ilog cplex 12.9.
\newblock 2019.

\bibitem{JainV01}
Kamal Jain and Vijay~V. Vazirani.
\newblock Approximation algorithms for metric facility location and $k$-median
  problems using the primal-dual schema and lagrangian relaxation.
\newblock {\em J. ACM}, 48(2):274--296, 2001.

\bibitem{joseph2016fairness}
Matthew Joseph, Michael Kearns, Jamie~H Morgenstern, and Aaron Roth.
\newblock Fairness in learning: Classic and contextual bandits.
\newblock In {\em Conference on Neural Information Processing Systems}, pages
  325--333, 2016.

\bibitem{kamishima2012fairness}
Toshihiro Kamishima, Shotaro Akaho, Hideki Asoh, and Jun Sakuma.
\newblock Fairness-aware classifier with prejudice remover regularizer.
\newblock In {\em Joint European Conference on Machine Learning and Knowledge
  Discovery in Databases}, pages 35--50, 2012.

\bibitem{khandani2010consumer}
Amir~E Khandani, Adlar~J Kim, and Andrew~W Lo.
\newblock Consumer credit-risk models via machine-learning algorithms.
\newblock {\em Journal of Banking \& Finance}, 34, 2010.

\bibitem{KS00}
Samir Khuller and Yoram~J Sussmann.
\newblock The capacitated k-center problem.
\newblock {\em SIAM Journal on Discrete Mathematics}, 13(3):403--418, 2000.

\bibitem{kiraly2012degree}
Tam{\'a}s Kir{\'a}ly, Lap~Chi Lau, and Mohit Singh.
\newblock Degree bounded matroids and submodular flows.
\newblock {\em Combinatorica}, 32(6):703--720, 2012.

\bibitem{KleinbergMR17}
Jon~M. Kleinberg, Sendhil Mullainathan, and Manish Raghavan.
\newblock Inherent trade-offs in the fair determination of risk scores.
\newblock In {\em Proc. 8th Conference on Innovations in Theoretical Computer
  Science}, pages 43:1--43:23, 2017.

\bibitem{Kleindessner2019}
Matth{\"{a}}us Kleindessner, Pranjal Awasthi, and Jamie Morgenstern.
\newblock Fair k-center clustering for data summarization.
\newblock In {\em Proc. 36th Proceedings, International Conference on Machine
  Learning (ICML)}, June 2019.

\bibitem{Kleindessner2019spectral}
Matth{\"{a}}us Kleindessner, Samira Samadi, Pranjal Awasthi, and Jamie
  Morgenstern.
\newblock Guarantees for spectral clustering with fairness constraints.
\newblock In {\em Proc. 36th Proceedings, International Conference on Machine
  Learning (ICML)}, June 2019.

\bibitem{census-paper}
Ron Kohavi.
\newblock Scaling up the accuracy of naive-bayes classifiers: A decision-tree
  hybrid.
\newblock In {\em Annual SIGKDD International Conference on Knowledge Discovery
  and Data Mining}, 1996.

\bibitem{LiS16}
Shi Li and Ola Svensson.
\newblock Approximating $k$-median via pseudo-approximation.
\newblock {\em SIAM J. Comput.}, 45(2):530--547, 2016.

\bibitem{luong2011k}
Binh~Thanh Luong, Salvatore Ruggieri, and Franco Turini.
\newblock k-nn as an implementation of situation testing for discrimination
  discovery and prevention.
\newblock In {\em Proc. 17th Annual SIGKDD International Conference on
  Knowledge Discovery and Data Mining}, pages 502--510, 2011.

\bibitem{malhotra2003evaluating}
Rashmi Malhotra and Davinder~K Malhotra.
\newblock Evaluating consumer loans using neural networks.
\newblock {\em Omega}, 31, 2003.

\bibitem{meek2002learningcurve}
Christopher Meek, Bo~Thiesson, and David Heckerman.
\newblock The learning-curve sampling method applied to model-based clustering.
\newblock {\em Journal of Machine Learning Research}, 2:397, 2002.

\bibitem{scikit-learn}
F.~Pedregosa, G.~Varoquaux, A.~Gramfort, V.~Michel, B.~Thirion, O.~Grisel,
  M.~Blondel, P.~Prettenhofer, R.~Weiss, V.~Dubourg, J.~Vanderplas, A.~Passos,
  D.~Cournapeau, M.~Brucher, M.~Perrot, and E.~Duchesnay.
\newblock Scikit-learn: Machine learning in {P}ython.
\newblock {\em Journal of Machine Learning Research}, 2011.

\bibitem{perlich2014machine}
Claudia Perlich, Brian Dalessandro, Troy Raeder, Ori Stitelman, and Foster
  Provost.
\newblock Machine learning for targeted display advertising: Transfer learning
  in action.
\newblock {\em Machine learning}, 95, 2014.

\bibitem{RS18}
Clemens R{\"o}sner and Melanie Schmidt.
\newblock {Privacy Preserving Clustering with Constraints}.
\newblock In {\em Proc. 45th International Colloquium on Automata, Languages
  and Programming}, pages 96:1--96:14, 2018.

\bibitem{schmidt2018fair}
Melanie Schmidt, Chris Schwiegelshohn, and Christian Sohler.
\newblock Fair coresets and streaming algorithms for fair k-means clustering.
\newblock {\em arXiv preprint arXiv:1812.10854}, 2018.

\bibitem{ShmoysT93}
David~B. Shmoys and {\'E}va Tardos.
\newblock An approximation algorithm for the generalized assignment problem.
\newblock {\em Mathematical programming}, 62(1-3):461--474, 1993.

\bibitem{Svitkina10}
Zoya Svitkina.
\newblock Lower-bounded facility location.
\newblock {\em ACM Trans. Alg.}, 6(4):69, 2010.

\bibitem{Swamy16}
Chaitanya Swamy.
\newblock Improved approximation algorithms for matroid and knapsack median
  problems and applications.
\newblock {\em ACM Trans. Alg.}, 12(4):49, 2016.

\bibitem{bank-paper}
Paulo~Rita Sérgio~Moro, Paulo~Cortez.
\newblock A data-driven approach to predict the success of bank telemarketing.
\newblock {\em Decision Support Systems}, 2014.

\bibitem{yang2017measuring}
Ke~Yang and Julia Stoyanovich.
\newblock Measuring fairness in ranked outputs.
\newblock In {\em Proc. 29th International Conference on Scientific and
  Statistical Database Management}, page~22. ACM, 2017.

\bibitem{ZafarVGG17}
Muhammad~Bilal Zafar, Isabel Valera, Manuel Gomez{-}Rodriguez, and Krishna~P.
  Gummadi.
\newblock Fairness constraints: Mechanisms for fair classification.
\newblock In {\em Proc. 20th Proceedings, International Conference on
  Artificial Intelligence and Statistics (AISTATS)}, pages 962--970, 2017.

\bibitem{zemel2013learning}
Rich Zemel, Yu~Wu, Kevin Swersky, Toni Pitassi, and Cynthia Dwork.
\newblock Learning fair representations.
\newblock In {\em Proceedings, International Conference on Machine Learning
  (ICML)}, pages 325--333, 2013.

\end{thebibliography}

\end{document}